\documentclass{comsoc2020}
\pdfoutput=1

\title{Perpetual Voting: The Axiomatic Lens}
\author{Martin Lackner and Jan Maly}

\usepackage{amssymb,latexsym}
\usepackage{graphicx,amsmath,amsfonts,amsthm}
\usepackage{comment}
\usepackage{tikz}
\usepackage{url}
\usepackage{mathrsfs}
\usepackage{paralist}
\usepackage[hidelinks]{hyperref}
\usepackage[]{units}
\usepackage{booktabs}
\usepackage{setspace}
\usepackage{xspace}
\allowdisplaybreaks
\usepackage{pifont}
\usepackage{booktabs}
\usepackage[numbers,sort&compress]{natbib}
\usepackage[inline]{enumitem}

\newtheorem{theorem}{Theorem}

\newtheorem{lemma}[theorem]{Lemma}

\newtheorem{definition}{Definition}

\newtheorem{proposition}[theorem]{Proposition}

\makeatletter
\newtheorem*{rep@theorem}{\rep@title}
\newcommand{\newreptheorem}[2]{
\newenvironment{rep#1}[1]{
 \def\rep@title{#2 \ref{##1}}
 \begin{rep@theorem}}
 {\end{rep@theorem}}}
\makeatother
\newreptheorem{theorem}{Theorem}
\newreptheorem{proposition}{Proposition}
\newreptheorem{lemma}{Lemma}
\newreptheorem{corollary}{Corollary}

\newcommand{\quota}{\mathit{qu}}
\newcommand{\satisfaction}{\mathit{sat}}

\newcommand{\phragmen}{Phragm\'{e}n}
\definecolor{OKgreen}{HTML}{035925}
\definecolor{NOred}{HTML}{8F061F}
\newcommand*\cmark{\textcolor{OKgreen}{\ding{51}}}
\newcommand*\xmark{\textcolor{NOred}{$\times$}}
\newcommand*\inftyno{\xmark}
\newcommand{\calR}{\mathcal{R}}
\newcommand{\calH}{\mathcal{H}}
\newcommand{\calD}{\mathcal{D}}
\newcommand{\newrule}{Exponential Rule}

\renewcommand{\leq}{\le}
\renewcommand{\geq}{\ge}
\newcommand{\score}{\mathit{sc}}

\begin{document}

\begin{abstract}
Perpetual voting was recently introduced as a framework for long-term collective decision making.
In this framework, we consider a sequence of subsequent approval-based elections and try to achieve a fair overall outcome.
To achieve fairness over time, perpetual voting rules take the history of previous decisions into account and
identify voters that were dissatisfied with previous decisions. In this paper, we look at perpetual voting rules from an axiomatic perspective and study two main questions. 
First, we ask how simple such rules can be while still meeting basic desiderata.
For two simple but natural classes, we fully characterize the axiomatic possibilities.
Second, we ask how proportionality
can be formalized in perpetual voting.
We study proportionality on simple profiles that are equivalent to the apportionment setting and
show that lower and upper quota axioms can be used to distinguish (and sometimes characterize) 
perpetual voting rules.
Furthermore, we show a surprising connection between a perpetual rule called Perpetual Consensus and Frege's apportionment method.
\end{abstract}

\section{Introduction}

In many voting scenarios, a group of voters, for example a committee  
or working group, has to make several decisions at
different points in time. If standard voting rules are used (such as Borda, plurality, etc.), it may happen 
that a majority dictates all decisions while some voters 
disagree with every outcome. This can lead to unrepresentative results and,
eventually, to dissatisfied voters dropping out of the voting process  \cite{debardeleben2009activating}.
Such situations are particularly undesirable if participation in the process 
is valued highly and if no extreme views are present in the electorate.
If a group of colleagues has to regularly agree on a meeting time, 
it is not acceptable if always the same colleague has to compromise.
Similarly, if a committee in a sports club is tasked with the organization of a 
Christmas party, no committee member should be completely ignored.

\emph{Perpetual voting}, recently introduced by Lackner \cite{Lackner2020}, is a formalism 
for tackling these types of long-term decision making processes.
From a formal point of view, a perpetual voting instance is a sequence of approval-based elections
where each decision has to be made `online', i.e., in the knowledge of past decisions
but without information about future elections.
Perpetual voting rules are deterministic, resolute functions that
take perpetual voting instances as input and that output a winning alternative for the current decision to be made.
Lackner \cite{Lackner2020} introduced several perpetual voting rules
that aim to achieve a fairer outcome over time as well as 
three basic axioms that formalize desirable properties in the perpetual setting:
(i) \emph{bounded dry spells} guarantee each voter a
satisfying outcome on a regular basis,
(ii) \emph{independence of uncontroversial decisions} states that unanimous decisions should
not impact other decisions, and (iii) \emph{simple proportionality} is a weak (minimal) proportionality requirement.
However, many axiomatic questions remain unanswered in the scope of perpetual voting. This paper aims to tackle two of them:

First, while Lackner \cite{Lackner2020} identified a rule that 
satisfies all three axioms, called Perpetual Consensus,
this rule requires a non-trivial calculation of weights and thus requires extra effort
to implement and explain in practice. In this paper, we investigate whether simpler 
voting rules can satisfy all -- or most of -- the basic axioms.
In order to do so, we define three simple sub-classes of 
perpetual voting rules: basic, win-based and loss-based 
weighted approval methods (WAMs). Win-based WAMs are
particularly interesting as they are a
natural adaption of the well-known sequential Thiele Methods from the 
multiwinner voting literature \cite{abcsurvey} to the perpetual voting
setting, We show that win-based and loss-based WAMs can satisfy at most 
one basic axiom at once.
In addition, we completely characterize the win- and loss-based WAMs that 
satisfy any of the basic axioms. Furthermore, we show that basic WAMs can satisfy
bounded dry-spells and independence of uncontroversial decisions at the same time. 
It remains open whether basic WAMs can even satisfy all three axioms.

Secondly, we investigate proportionality measures that strengthen the minimal requirement
of simple proportionality. 
In particular, we study the behavior of the perpetual voting rules
on simple decision sequences that are equivalent to the apportionment setting.
We show that, even in this very restricted setting, only few perpetual voting rules
show a proportional behavior. In particular, we show that Perpetual PAV 
is the only win-based WAM that satisfies apportionment lower quota.
In contrast, loss-based WAMs cannot satisfy even simple proportionality.
Moreover, we explore which apportionment methods arise if perpetual rules
are restricted to the apportionment setting.
Interestingly, Perpetual Consensus equals Frege's apportionment method (discussed below in the related work).
Finally, we discuss the possibility
of stronger proportionality axioms and introduce the notion of a 
perpetual proportionality degree.

\subsubsection*{Related work}

Perpetual voting is related to other approaches that consider either temporal aspects of voting or sequences of decisions.
We give a brief overview of the closest relatives.
\citet{Freeman2017Fair} proposed a sequential mechanism for the aggregation of utility functions over time with the goal to maximize long-term Nash welfare.
Perpetual voting does not assume utility functions (but approval-based preferences) and thus cannot rely on utility-based concepts such as Nash welfare.
Offline (non-sequential) variants of this formalism have been studied by \citet{conitzer2017fair} and \citet{freeman2018dynamic}.

The \emph{storable votes} method~\citep{casella2005storable,casella2012storable} is a voting rule that is applicable in the temporal setting of perpetual voting.
However, with storable votes, voters are required to strategically store and spend votes to achieve a fair outcome; this in contrast to perpetual voting rules where fairness (proportionality) should be achieved without strategic voting.

\emph{Sequential voting rules} \citep{lang2009sequential} work on combinatorial domains \citep{Lang2016VotinginCombinatorial} for which a natural order of issues exists.
Due to this order, paradoxes disappear that can occur in arbitrary combinatorial domains.
Perpetual voting and sequential voting differ in key aspects, e.g., 
perpetual voting rules depend on the history of previous decisions (this is not the case for sequential voting rules) and sequential voting was introduced to represent preferences in combinatorial domains (this is not the goal of perpetual voting).

\emph{Frege's method}~\cite{frege:2000a,frege}, proposed by the philosopher Gottlob Frege, is a voting rule with a distinct temporal component, but -- in contrast to perpetual voting -- 
is not intended to grant fairness to voters. Instead, candidates accumulate weight and thus Frege's method achieves outcomes that are fair for candidates. A modified version of Frege's method gives rise to
an apportionment method~\cite{frege} that is related to Perpetual Consensus (see Section~\ref{sec:lower-upper-quota}).

\section{The Perpetual Voting Framework}\label{sec:pv-preliminaries}

We will now introduce the perpetual voting formalism, as defined in \cite{Lackner2020},
alongside necessary basic definitions.
Let $N=\{1,\dots, n\}$ be a set of voters (agents). 
Given a set of alternatives~$C$, we assume that each voter $v\in N$ approves some non-empty subset of~$C$. An \emph{approval profile} $A=(A(1),\dots,A(n))$ for~$C$ is an $n$-tuple of subsets of $C$, i.e., $A(v)\subseteq C$ for $v\in N$.
We call the triple $(N, A, C)$ a \emph{decision instance}.

A \emph{$k$-decision sequence} $\calD=(N,\bar A, \bar C)$ is a triple consisting of a set of voters $N$, a $k$-tuple of approval profiles $\bar A=(A_1,A_2,\dots, A_{k})$ and an associated $k$-tuple of alternatives $\bar C=(C_1,\dots,C_{k})$.
Thus, for each $i\leq k$, the triple $(N,A_i,C_i)$ is a decision instance and can be seen as an individual decision to be made. 

We write $\bar w \in \bar C$ to denote a $k$-tuple $\bar w=(w_1,\dots, w_k)$ with $w_i\in C_i$ for $i\in \{1,\dots, k\}$; we refer to $\bar w$ as a \emph{$k$-choice sequence}.
This tuple represents the chosen alternatives in rounds $1$ to $k$.
We write $\bar w_{\leq t}$ to denote the t-tuple $(w_1,\dots, w_t)$.
If we combine a $k$-decision sequence $(N,\bar A, \bar C)$ and a $k$-choice sequence $\bar w \in \bar C$, we speak of a \emph{$k$-decision history} $\calH=(N,\bar A, \bar C, \bar w)$, which can be seen as the history of past decisions alongside the made choices.
We thus know, for any $i\leq k$, that in case of decision instance $(N,A_i,C_i)$ alternative $w_i$ was chosen. 

An important statistic of $k$-decision histories is the satisfaction of each voter:
Given a $k$-decision history $\calH=(N, \bar A, \bar C, \bar w)$,
the \emph{satisfaction} of voter $v\in N$ with $\bar w$ is $\satisfaction(v, \bar w)=|\{i\leq k: w_i\in A_i(v)\}|$.
Thus, the satisfaction of a voter is the number of past decisions that have satisfied this voter.
Note that although satisfaction 
clearly depends on $\calH$, we do not explicitly mention that in the notation (and other definitions throughout the paper).

Assume that a group of voters $N$ wants to take a decision and looks back at $k$~decisions already taken.
That is, we are presented with a $k$-decision history $\calH=(N,\bar A, \bar C, \bar w)$ and a decision instance $(N, A_{k+1}, C_{k+1})$.
The question now is which alternative in $C_{k+1}$ should be chosen, subject to the preferences in $A_{k+1}$ and under consideration of $\calH$.
An \emph{(approval-based) perpetual voting rule~$\calR$} is a function that maps a pair of a decision instance $(N, A_{k+1}, C_{k+1})$ and a $k$-decision history $\calH$ to an alternative in $C_{k+1}$.

Given a $k$-decision sequence $\calD=(N,\bar A, \bar C)$, we write $\calR(\calD)$ to denote the $k$-choice sequence $\bar w\in \bar C$ as selected by the perpetual voting rule $\calR$, that is, $\calR(\calD)=\bar w$ is inductively defined by $w_i=\calR(N, (A_1,\dots, A_i), (C_1, \dots, C_{i}), (w_1,\dots,w_{i-1}))$ for $i\leq k$.
We expect perpetual voting rules to be resolute, i.e., return exactly one winning alternative,
therefore we require a tie-breaking order to resolve ties.
Throughout the paper, we assume that there exists some arbitrary and fixed order for each set of alternatives that settles ties.

\subsubsection*{Weighted Approval Methods}

A natural approach to define perpetual voting rules is via weights: voters that have been previously neglected receive a higher weight, voters that are already very satisfied receive a lower weight.
In each round, the alternative that receives the highest sum of weighted approvals is selected.
This idea is captured in a broad sense by the class of \emph{perpetual weighted approval methods\footnote{
We note that WAMs in this paper are defined slightly more general than in \cite{Lackner2020}.} (WAMs)},
which contains most rules proposed in \cite{Lackner2020}.
These approval-based perpetual voting rules are defined as follows: Each voter has an assigned positive weight, which may change each round;
a larger weight corresponds to being assigned a higher importance.
Let $\alpha_k(v)$ denote voter $v$'s weight in round $k$. Weights are initialized with $\alpha_1(v)=1$ for all $v\in N$.
The weights of voters in the following rounds are a consequence of the previous history.
Formally, 
there exists a weight function $h$ such that for all $v\in N$,
$\alpha_{k+1}(v) = h(v, \calH)$.
Given a $k$-decision history $\calH = (N, \bar A, \bar C, \bar w)$ and a decision instance $(N,A_{k+1},C_{k+1})$,
the rule selects an alternative $w_{k+1}\in C_{k+1}$ with maximum weighted approval score. 
That is, 
the score of an alternative $c$ is defined as
\[\score_{k+1}(c) = \sum_{v\in N\text{ with } c\in A_{k+1}(v)}\alpha_{k+1}(v).\]

As we allow arbitrary functions for the computation of the weights,
WAMs can be extremely complex. However,
simple voting rules are often preferable in (perpetual) voting,
as they make it easier for voters to understand why an alternative was chosen.
Therefore, in the next section, we consider simpler forms of WAMs.

\section{Basic Weighted Approval Methods}\label{sec:basicWAMs}

In this section, we study three simple subclasses of WAMs: \emph{basic},
\emph{win-based} and \emph{loss-based WAMs}.
For basic WAMs, the weights of voters are calculated in a straight-forward fashion: a voter's weight depends only on the voter's weight in the previous round and whether the voter was satisfied with the previous decision.
\begin{definition}
A perpetual weighted approval method (WAM) is a \emph{basic WAM} if the weights of voters $v\in N$ can be computed via two functions $f$ and $g$ as follows:
\[\alpha_{k+1}(v) = \begin{cases} f(\alpha_k(v)) & \text{if }w\notin A_{k}(v),\\
g(\alpha_k(v)) & \text{if }w\in A_{k}(v).\end{cases}\]
We further require that $g(x) \leq x$ and $f(x) \geq x$ for all inputs $x$.
\end{definition}
The condition that $g(x) \leq x$ and $f(x) \geq x$ for all inputs is natural because being satisfied with a decision should not result in a larger future weight and 
being unsatisfied should not lead to a smaller weight in the future.
Note also that for basic WAMs the weight of a voter does not depend on the weights of other voters nor does it depend on the total number of voters
(the functions $f$ and $g$ are the same for all $N$).
Win- and loss-based WAMs are even simpler subclasses of basic WAMs
\begin{definition}
We call a basic WAM a win-based WAM if $f(x) = x$ and a loss-based WAM if $g(x) = x$.
\end{definition}
By definition, the weight of a voter in a win-based WAM only depends on how many rounds
she has already won (i.e., how many rounds she was satisfied with). Similarly, the weight of a voter in a loss-based WAM only depends
on how many rounds she already lost.
In particular, this means a win- or loss-based WAM is fully defined by a sequence $(w_1, w_2 \dots w_i, \dots)$
such that $w_i$ is the weight of a voter that has won resp.\ lost $i$ rounds.
Win-based WAMs are closely related to 
the (sequential) Thiele methods used in multi-winner voting~\cite{abcsurvey};
loss-based WAMs are related to dissatisfaction counting rules~\cite{ijcai/LacknerSkowron-multiwinner-strategyproofness}.
We consider two of the most common sequential Thiele methods
as win-based WAMs:

\begin{description}
\item[AV.] 
The simplest example of a WAM is approval voting (AV), which
completely ignores the history of past decisions.
AV corresponds to the basic WAM with $f(x) = g(x) = x$.

\item[Perpetual PAV.] This method is inspired by Proportional Approval Voting~\citep{Thie95a,Faliszewski2017MultiwinnerVoting:A,abcsurvey}, or more specifically by its sequential counterpart, and is thus based on the harmonic series.
The weight of voters is defined by
\begin{align*}
\alpha_{k+1}(v) &= \frac{1}{\satisfaction(v, \bar w)+1}= \begin{cases} \alpha_k(v) & \text{if }w_k\notin A_{k}(v),\\
\frac{\alpha_k(v)}{\alpha_k(v)+1} & \text{if }w_k\in A_{k}(v).\end{cases}.
\end{align*}
The last equality shows that Perpetual PAV is indeed a win-based WAM.

\end{description}

\noindent
Observe that AV is win-based and loss-based at the same time.
Moreover, the rule Perpetual Unit Cost considered in \cite{Lackner2020}
is a loss-based WAM. However, as this rule does not show a particularly
promising behaviour, we omit its definition here.
Finally, the following is an example for a basic WAM that is neither win- nor loss-based.

\begin{description}
\item[Perpetual Reset.] The weight of a voter is increased by one
if he is not satisfied with the winning alternative and reset to $1$ if 
he is satisfied, i.e., $f(x) = x +1$ and $g(x) = 1$.
\end{description}

Due to their sequential nature, all aforementioned rules can be computed in polynomial time.
We want to investigate if such simple voting rules can satisfy the requirements we have for
a fair and robust perpetual voting rule. In order to do so,
we consider the three basic axioms that were introduced in \cite{Lackner2020}.

\subsubsection*{Bounded Dry Spells}

The \emph{bounded dry spells} property guarantees that every voter is satisfied with at least one choice in a bounded number of rounds.
This property is very important for creating an incentive to participate in the decision making process.

\begin{definition}[Dry spells]
 Given a $k$-decision history $\calH=(N, \bar A,\bar C,\bar w)$, we say that a voter $v\in N$ has a \emph{dry spell of length $\ell$} if there exists $t\leq k-\ell$ such that $\satisfaction(v, \bar w_{\leq{t}})=\satisfaction(v, \bar w_{\leq t+\ell})$, i.e., voter $v$ is not satisfied by any choice in rounds $t+1,\dots,t+\ell$.

Let $d$ be a function from $\mathbb{N}$ to $\mathbb{N}$.
A perpetual voting rule $\calR$
has a \emph{dry spell guarantee of $d$} if for any decision sequence
$\calD = (N, \bar A,\bar C)$ and $\bar w = \calR(\calD)$, no voter has a dry spell of length $d(|N|)$.
A perpetual voting rule $\calR$ has \emph{bounded dry spells} if $\calR$ has a dry spell guarantee of some $d$.
\end{definition}
As win- and loss-based WAMs only consider the number of wins resp.\ losses but not the round in which they 
occur, a long winning streak can be followed by a arbitrarily long dry spell.

\newcommand{\notconst}{Every win-based and loss-based WAM has unbounded dry spells.}
\begin{proposition}
\notconst\label{prop:notconst}
\end{proposition}

However, there are basic WAMs with bounded dry spells.

\newcommand{\setzerodsg}{Perpetual Reset has a dry spell guarantee of $2n-2$. This guarantee is tight, i.e., dry spells of length $2n-3$ may occur.}
\begin{proposition}[Lackner \cite{Lackner2020}]
\setzerodsg
\end{proposition}

\newcommand{\phragdsg}{Perpetual \phragmen{} has a dry spell guarantee of $2n-1$. This bound is tight.}

\newcommand{\consensusdsg}{Perpetual Consensus has a dry spell guarantee of at least $\frac{n^2+3n}{4}$.}

\newcommand{\dryspellinfty}{AV and Perpetual PAV 
have unbounded dry spells.}

\subsubsection*{Independence of Uncontroversial Decisions}

The second axiom concerns the impact of uncontroversial decisions, i.e., decisions where a choice can be made with unanimous agreement. 
Such uncontroversial decisions should not have an impact on other decisions, as otherwise the inclusion of such decisions could be used to manipulate the outcome of the decision process.
Formally, an approval profile $A$ is \emph{uncontroversial due to $c$} if $\bigcap_{v\in N}A(v) = \{c\}$.
Furthermore, given a $k$-tuple $L=(l_1,\dots,l_k)$ and $i\in\{0,\dots,k\}$, we write 
$L \oplus_i x$ to denote the $(k+1)$-tuple $L=(l_1,\dots,l_i, x, l_{i+1},\dots,l_k)$.

\begin{definition}[Independence of uncontroversial decisions]
A perpetual voting rule $\mathcal{R}$
satisfies \emph{independence of uncontroversial decisions} if for any $k$-decision sequence $(N, \bar A,\bar C)$, approval profile $A$ for $C$ that is uncontroversial due to $c$, and $i\in\{0,\dots,k\}$
it holds that 
\[\mathcal{R}(N, \bar A \oplus_i A,\bar C \oplus_i C) = \mathcal{R}(N, \bar A,\bar C)\oplus_i c.\]
\end{definition}

\newcommand{\propiud}{Perpetual PAV, 
Perpetual Reset, Rotating Dictator, and Perpetual \phragmen{} fail independence of uncontroversial decisions.}
\newcommand{\thmiud}{AV 
and Perpetual Consensus satisfy independence of uncontroversial decisions}
Clearly AV -- and any other loss-based WAM -- satisfies independence of uncontroversial 
decisions by definition, as the weight of a satisfied voter does not change.
It turns out that all basic WAMs that satisfy independence of uncontroversial 
decisions can only change the weight of a winning voter by multiplying it with a constant,
though this constant may change from round to round.
 
\newcommand{\IUDident}{Let $\calR$ be a basic WAM. 
Then, $\calR$ satisfies independence 
of uncontroversial decisions if and only if for every $k$ there is a constant $c_k$ such that
for every $x$ that can occur as a weight after $k$-rounds we have 
$g(x) = c_k x$ for $c_k > 0$.}
\begin{proposition}\label{prop:IUDident}
\IUDident
\end{proposition}

While loss-based WAMs cannot have bounded dry spells, the 
additional flexibility of a varying constant factor allows us to define
a basic WAM that has bounded dry spells and satisfies independence of uncontroversial decision.
Although this rule belongs in the class of basic WAMs, it is far from intuitive or simple.

\newcommand{\newruleDef}{
\begin{description}
\item[\newrule] Let $r(x)$ be the function that maps $x$ to $\frac{m+1}{2}$
where $m$ is the smallest integer such that there is a $l \geq 0$ for which $2^l x = m$. 
If no such $m$ exists, then $r(x)$ is undefined (these values are never required).
This means in particular that if $x$ is of the form $\frac{2k-1}{2^\ell}$ for integers $k$ and $\ell\geq 0$, then $r(x)=k$.
Then, \newrule{} is defined by 
\begin{align*}
f(x) = \left(\frac{2r(x)+1}{2r(x) -1}\right) x\quad \text{ and }\quad 
g(x) = \left(\frac{2r(x)+1}{2^{r(x)!}(2r(x) -1)}\right) x
\end{align*}
\end{description}}

\newruleDef

The weight functions $f$ and $g$ are chosen such that the weight of a voter will always be 
a fraction in which the numerator encodes which round we are in while the denominator
guarantees that voters with maximal dry spell are more influential than all other voters.
In particular, the numerator in round $k$ will be the $k$-th odd number for all voters,
while the denominator will always be of the form $2^l$ for some $l \in \mathbb{N}$, thus $r(x)$ computes the 
unique round in which the weight $x$ can appear.
Moreover, the ``penalty" for winning, i.e., $2^{r(x)!}$ is chosen such that, 
for every $k$, a voter who has won in round $1, \dots ,k$ has a higher weight than a voter 
who won in round $k+1$, which ensures bounded dry spells.

\newcommand{\newruleBD}{The \newrule{} has bounded dry spells and satisfies independence of uncontroversial decisions.}
\begin{proposition}\label{prop:newruleBD}
\newruleBD
\end{proposition}

\subsubsection*{Simple Proportionality}

Voting rules with bounded dry spells can guarantee every voter some form of representation. 
An alternative approach is to look at groups of voters that have identical preferences and 
guarantee them a \emph{proportional} representation.

\begin{definition}[Simple proportionality]\label{def:sp}
We say that a $k$-decision sequence $\calD = (N, \bar A,\bar C)$ 
is simple if $A_1=\dots=A_k$, $C_1=\dots=C_k$, and $|A_1(v)|=1$ for all $v\in N$.
Given a simple decision sequence $\calD$ and a voter $v\in N$, let $\#v$ denote the number of voters with identical preferences, i.e., $\#v=|\{v'\in N: A(v')=A(v)\}|$.
A perpetual voting rule $\calR$ satisfies \emph{simple proportionality} if
for any simple $n$-decision sequence $\calD$ with $|N|=n$ 
it holds that $\satisfaction(v, \calR(\calD))= \#v$ for every voter $v\in N$.
\end{definition}

Although it is a rather weak proportionality requirement
(similar to weak proportionality in the apportionment setting~\citep{Balinski1982FairRepresentation:Meeting}),
it is sufficiently strong to reveal that some perpetual voting rules are not proportional.
Indeed, AV, 
Perpetual Reset and \newrule{} 
fail simple proportionality. On the other hand, Perpetual PAV 
satisfies simple proportionality.
(See \cite{Lackner2020}, Proposition~\ref{newruleFailSP} in the appendix and Appendix~\ref{App:orig-proofs} for proofs.)
Perpetual PAV witnesses that win-based WAMs can satisfy simple proportionality.
Loss-based WAMs, on the other hand, never satisfy simple proportionality.

\newcommand{\failsp}{AV 
and Perpetual Reset 
fail simple proportionality.}

\newcommand{\satisfysp}{Perpetual PAV, Perpetual Consensus, 
and Rotating Dictator
satisfy simple proportionality.
}

\newcommand{\SPropIUD}{No loss-based WAM satisfies 
simple proportionality.}
\begin{proposition}
\SPropIUD\label{prop:SPropIUD}
\end{proposition}

\begin{proof}[Proof sketch.]
Assume for the sake of a contradiction that $\calR$ is a loss-based WAM 
that satisfies simple proportionality.
Now, consider for an arbitrary $\ell \geq 1$ a $\ell+2$-decision sequence $(N,\bar A, \bar C)$ such that 
$N = \{v_0, \dots, v_{\ell+1}\}$,
$C_1 = \dots = C_{|N|} =  \{a,b\}$ and $v_0$ always votes $\{a\}$ and $v_i$ votes $\{b\}$
for all $i \in \{1,\dots, \ell +1\}$. Because $\calR$ satisfies simple proportionality
there are two possible cases in round $\ell+2$:
either $a$ has won one or zero times. In the first case, simple proportionality implies
\begin{equation}
 f^\ell(1) \leq (\ell+1) f(1) \label{eq:loss-sp-1}.\\ 
\end{equation}
In the second case, we must have $f^{i-1}(1) \leq \ell +1$ for all $i \leq \ell+1$, hence
\begin{equation}
 f^\ell(1) \leq \ell+1  \label{eq:loss-sp-2}
 \end{equation} 
 
Now consider a second $k+1$-decision sequence $(N,\bar A, \bar C)$ such that 
$N = \{v_0, \dots, v_k\}$, $C_1 = \dots = C_{|N|} =  \{a,b\}$,
$v_0, v_1$ always vote $\{a\}$ and $v_i$ votes $\{b\}$
for all $i \in \{2,\dots, k\}$. Then, there must be a round $i$ where 
$a$ wins the second time. In this round, the score of $b$ is $(k-1)f(1)$
and the score of $a$ is at most $2\cdot f^{k-1}(1)$.
As we know that $a$ wins in round $i$ we have 
\begin{align}
\nicefrac{1}{2}(k-1)f(1)\leq f^{k-1}(1).\label{eq:loss-sp-3}
\end{align}

Finally, consider a $2k$-decision sequence $(N,\bar A, \bar C)$ such that 
$N = \{v_1, \dots, v_k, w_1, \dots ,w_k\}$, $C_1 = \dots = C_{|N|} =  \{a,b_1, \dots, b_k\}$ and
$v_i$ always votes $\{a\}$ and $w_i$ votes $\{b_i\}$
for all $i \in \{1,\dots, k\}$.
Furthermore, assume w.l.o.g.\ that a tie-breaking is applied that always picks $b_i$ over $b_j$ if $i < j$.
Thus, $b_i$ for all $i \leq k-1$ must win before $b_k$. 
We claim that $b_k$ does not win any of the~$2k$ first rounds.
Let us consider any round in which all $b_i$ for $i \leq k-1$ have already won, hence $b_k$ could win.
Then, by Equation~\eqref{eq:loss-sp-3}, the score of $a$ is $k \cdot f^{k-1}(1) \geq \nicefrac{1}{2}\, k(k-1)f(1)$,
while the score of $b_k$ is at most $f^{2k-1}(1)$.

We distinguish whether Equation~\eqref{eq:loss-sp-1} or~\eqref{eq:loss-sp-2} holds.
If~\eqref{eq:loss-sp-1} holds (for $\ell=2k-1$), the score of $b_k$ is $f^{2k-1}(1) \leq 2kf(1)$, which is at less than the score of $a$ for sufficiently large $k$.
Hence, $b_k$ does not win and $\calR$ does not satisfy simple proportionality.
If~\eqref{eq:loss-sp-2} holds (for $\ell=2k-1$), the score of $b_k$ is $f^{2k-1}(1) \leq  2k$.
Furthermore, we know $f(1) \geq 1$.
Now, for any $k$ large enough $2k <\nicefrac{1}{2}(k-1)k \leq \nicefrac{1}{2}(k-1)kf(1)$.
Hence, $b_k$ does not win and $\calR$ does not satisfy simple proportionality.
\end{proof}

Moreover, even for win-based WAMs, the class of rules that satisfies simple proportionality
is quite restricted, as the following classification result shows.

\newcommand{\simplepropChar}{Let $\mathcal{R}$ be a win-based WAM.
Furthermore, define the sequence $w= (1, g(1), g(g(1)), g(g(g(1))), \dots)$.
Then, $\mathcal{R}$ satisfies simple proportionality if and only if 
$xw_x<(y+1)w_y$ for all integers $x,y\geq 0$.}
\begin{theorem}\label{simplepropChar}
\simplepropChar
\end{theorem}

Examples of such rules include PAV with $w=(1, \nicefrac 1 2 , \nicefrac 1 3 , \dots)$ but also,
for example, $w=(1, \nicefrac{1}{c+1} , \nicefrac{1}{2c+1} , \dots)$ for all $\nicefrac 1 2 <c \leq 1$.
From this and Proposition~\ref{prop:IUDident}, we can conclude that there
is no win-based WAM that satisfies independence of uncontroversial decisions
and simple proportionality at the same time.

\begin{table}
\begin{center}
\begin{tabular}{lccccccc}
                     & BD        & IUD      & SP       &  BD+IUD   & BD+SP  & IUD+SP & BD+IUD+SP\\
\midrule                   
win-based WAMs        & \xmark    & \cmark   & \cmark   &  \xmark     & \xmark   & \xmark   &  \xmark     \\
loss-based WAMs       & \xmark    & \cmark   & \xmark   &  \xmark     & \xmark   & \xmark   &  \xmark     \\
basic WAMs            & \cmark    & \cmark   & \cmark   &  \cmark     & ?        & ?        &  ?          \\
WAMs                  & \cmark    & \cmark   & \cmark   &  \cmark     & \cmark   & \cmark   &  \cmark     \\
\end{tabular}
\caption{An overview of classes of perpetual voting rules and the axiomatic analysis with respect to Bounded Dry Spells (BD), Independence of Uncontroversial Decisions (IUD), and Simple Proportionality (SP).}
\label{tab:overview2}
\end{center}
\end{table}

\newcommand{\winbasedspropIUD}{No win-based WAM
satisfies simple proportionality and independence of uncontroversial decisions
at the same time.}
\begin{proposition}\label{prop:winbasedspropIUD}
\winbasedspropIUD
\end{proposition}

Table~\ref{tab:overview2} summarizes the results of this section.
We conclude that win- and loss-based WAMs are too simple to satisfy
all three desiderata for a perpetual voting rule. On the other hand, it is still open
whether basic WAMs can satisfy all desiderata we considered in this section.
However, the \newrule{} shows that even basic WAMs can already be quite complex
and unintuitive. 
Hence, the presented results point towards an incompatibility between simplicity and the satisfaction of all three considered axioms.

As Table~\ref{tab:overview2} shows, there exist WAMs that satisfy all three axioms discussed in this section. Perpetual Consensus is such a rule (introduced in \cite{Lackner2020}; we define it in the following section).
Our next goal is to analyze which rules satisfy stronger 
proportionality requirements than simple proportionality.

\section{A Closer Look at Proportionality}\label{sec:proportionality}

We now want to 
figure out which 
perpetual voting rules can be deemed proportional; simple proportionality---as the name suggest---might be a too weak requirement for this.
Indeed, consider the following perpetual rule:

\begin{description}
\item[Rotating Dictator.] Let $(v_0, \dots v_{n-1})$ be an enumeration of the voters.
In round $k$, voter $v_{(k\mod n)}$ is a dictator, i.e., an alternative is selected from $A_k(v_{(k \mod n)})$ according to some tie-breaking rule. This rule is a WAM since the weight of the dictator can be set to 1, all others to 0.
\end{description}
 
This rule is certainly not proportional in a strong sense, because in every round,
it only takes the preferences of one voter into account. Nevertheless it can be checked that
it satisfies simple proportionality (see Proposition~\ref{RotDictSP} in the appendix).

Before we define stronger notions of proportionality,
let us first introduce two further perpetual rules:
Perpetual Consensus, introduced in \cite{Lackner2020}, 
and Perpetual \phragmen{}, a new rule based on \phragmen's sequential rule.
As we will see, both of them can be viewed as proportional---but each in a different sense.

\begin{description}
\item[Perpetual Consensus.]
This WAM is based on the idea that the weight of voters that are satisfied with a decision is reduced in total by $n$ and this number is divided equally among them.
Consequently, voters can have negative weights\footnotemark; voters with negative weights are not taken into account when determining the winning alternative.
After each decision, the weight of all voters is increased by $1$. 
Formally,
let $N_k^+(c) = \left\{ v\in N : c\in A_{k}(v) \text{ and } \alpha_{k}(v)>0\right\}$,
and for all $v\in N$, $\alpha_1(v)=1$  and
\[\alpha_{k+1}(v) = \begin{cases} \alpha_k(v)+1 & \text{if }v\notin N_{k}^+(w_k),\\
\alpha_k(v)+1-\frac{n}{|N_k^+(w_k)|} & \text{if }v\in N_{k}^+(w_k).\end{cases}\]
Thus, the  score of an alternative $c$ is defined as
\[\score_{k+1}(c) = \sum_{v\in N_{k+1}(c)}\max(0,\alpha_{k+1}(v)).\]

\footnotetext{
Although this definition does not quite fit the definition of WAMs because of negative weights,
it can easily be adapted to that framework by defining a voting rule that assigns the
same weights as Perpetual Consensus if $\alpha_k(v)$ is positive and $0$ otherwise.
}

\item[Perpetual \phragmen.] This rule is an adaption of \phragmen's Sequential Rule \citep{Phra94a,Brill2017Phragmens}. It can be described as a load distribution procedure. We assume that winning a round incurs a load of $1$, which is distributed to a set of voters that jointly approve the winning alternative. Let $\ell_k(v)$ denote the load assigned to voter~$v$ in rounds~$1$ to~$k$ ($\ell_1(v) = 0$).
In round $k+1$,
for each set of voters $N'$ that jointly approves at least one alternative ($\bigcap_{v\in N'} A_{k+1}(v)\neq \emptyset$), we calculate 
\[\ell_{k+1}(N')=\frac{1 + \sum_{v\in N'} \ell_k(v)}{|N'|};\]
this is the load that each of these voters would have if they were selected to choose the winning alternative.
As Perpetual \phragmen{} aims at keeping voters' loads as small as possible, 
the set of voters $N'$ is selected where $\ell_{k+1}(N')$ is minimal.
If the intersection of their approval sets contains more than one alternative, a tie-breaking rule is used. If more than one set of voters exists with minimal $\ell_{k+1}(N')$, then the set of voters according to another arbitrary tie-breaking order is chosen. 
Let $N'$ be the set of voters selected to choose an alternative.
Then the loads in the next round are defined as
\[\ell_{k+1}(v) = \begin{cases} \ell_k(v) & \text{if }v\notin N',\\
\ell_{k+1}(N') & \text{if }v\in N'.\end{cases}\]

\end{description}

Perpetual \phragmen{} has the following basic properties.

\newcommand{\phragmenprops}{
Perpetual \phragmen{} 
\begin{enumerate*}[label=(\roman*)]
\item satisfies simple proportionality,
\item has a dry spell guarantee of $2n-1$ (this bound is tight),
\item fails independence of uncontroversial decisions,
\item is not equivalent to any WAM, and \label{phragWAM}
\item is computable in polynomial time.
\end{enumerate*}}

\begin{proposition}\label{prop:phragmen-props} \phragmenprops
\end{proposition}

\begin{proof}[Proof of \ref{phragWAM}]
We consider a decision sequence with 7 voters.
In the first two rounds, the preferences are \[(\{a\},\{a\},\{a\},\{b\},\{b\},\{b\},\{c\}).\]
Thus, $a$ wins in the first round (we assume alphabetic tie-breaking) and the corresponding loads are $(\frac 1 3,\frac 1 3,\frac 1 3,0,0,0,0)$.
In the second round $b$ wins and the loads are $(\frac 1 3,\dots,\frac 1 3,0)$.
Assume towards a contradiction that Perpetual \phragmen{} is a WAM. 
We can thus assign weights in some fashion; let these be $x_1,\dots,x_7$.

We now consider several decision instance for round three.
First, if the preferences are \[(\{a\},\{b\},\{c\},\{d\},\{e\},\{f\},\{a,b,c,d,e,f\}),\] then all alternatives are tied.
Thus, we can conclude that $x_1=x_2=x_3=x_4=x_5=x_6$; let $x=x_1=\dots=x_6$ and $y=x_7$.
Second, if the preferences are \[(\{a\},\{a\},\{b\},\{c\},\{d\},\{e\},\{f\}),\] then $a$ wins (the load of voter 1 and 2 would increase to $\frac 5 6$).
Thus, we infer that $2x > y$.
Finally, we consider \[(\{a\},\{a\},\{a\},\{b\},\{c\},\{c\}).\]
Here $a$ and $c$ are tied; in both cases the load of the corresponding voters would increase to $\frac 2 3$.
Thus, it holds that $3x = x+y$ and in turn $2x = y$. This contradicts our previous result that $2x > y$.
We conclude that Perpetual \phragmen{} cannot be ``simulated'' by a WAM.
\end{proof}

\subsection{Upper and Lower Quota in the Apportionment Setting}\label{sec:lower-upper-quota}

Ideally, we would like to define proportionality on arbitrary instances.
Unfortunately, Lackner \cite{Lackner2020} has shown that there are instances
in which every choice sequence is unproportional in a strong sense.
Instead, we enlarge the class of decision sequences by allowing simple decision sequences of arbitrary length.
This coincides with the apportionment setting \cite{Balinski1982FairRepresentation:Meeting,Brill2017MultiwinnerApprovalRules}. 
Apportionment is the problem of assigning seats to political parties in parliamentary elections. This setting is directly captured by simple profiles (Definition~\ref{def:sp}): voters approve a single alternative (this alternative corresponds to a party) and each round one party is selected (this corresponds to one party being assigned a seat).
We consider two axioms:

\begin{definition}[Apportionment lower/upper quota]
A perpetual voting rule $\calR$ satisfies \emph{apportionment lower quota (ALQ)} if for any simple decision sequence $\calD$, it holds that 
for every voter $v\in N$, \[\satisfaction(v, \calR(\calD))\geq\left\lfloor{k\cdot \frac{\#v}{n} }\right\rfloor.\]
A perpetual voting rule $\calR$ satisfies \emph{apportionment upper quota (AUQ)} if for any simple decision sequence $\calD$, it holds that 
for every voter $v\in N$, $\satisfaction(v, \calR(\calD))\leq\left\lceil{k\cdot \frac{\#v}{n} }\right\rceil$.
\end{definition}

This is still a weak notion of proportionality, but it suffices to show that Rotating Dictatorship is not proportional
as it does not satisfy either ALQ or AUQ (see Proposition~\ref{RotDictALQ} in the appendix). Moreover,
with these stronger axioms, we can extend Theorem~\ref{simplepropChar} to a full
characterization of Perpetual PAV among the win-based WAMs.

\newcommand{\pavchar}{Let $\calR$ be a win-based WAM that satisfies ALQ.
Then $\calR$ must be equivalent to Perpetual PAV. The class of basic WAMs contains additional rules satisfying ALQ.}
\begin{theorem}
\pavchar\label{thm:pav-char}
\end{theorem}

The proof of the first statement is a straightforward adaption of a similar proof (based on extended justified representation) by \citet{Aziz2017Justifiedrepresentationin}.
The second part, that the characterization does not extend to basic WAMs, is due to the following example.
Consider the basic WAM $\calR_f$ defined by:
\begin{align*}
f(x)&= \begin{cases} \pi\cdot x & \text{if $x$ is a rational number},\\
x & \text{otherwise}.\end{cases}  \quad\ \ 
g(x)&= 
\begin{cases} \frac{1}{\satisfaction{} + 1} & \text{if $x$ is a rational number},\\
\frac{\pi}{\satisfaction{} + 1} & \text{otherwise}.\end{cases}
\end{align*}
This function essentially resembles PAV but the weight of a voter is multiplied by $\pi$ the first time she is unsatisfied with a decision. This change does not impact ALQ.

Only very few of the considered voting rules satisfy ALQ or AUQ (Table~\ref{tab:overview} provides an overview).
First of all, ALQ and AUQ imply simple proportionality. Therefore, 
only rules that satisfy simple proportionality can satisfy ALQ or AUQ.
We prove the following result:

\newcommand{\propapportionment}{Perpetual PAV and Perpetual \phragmen{}, when restricted to the apportionment setting, correspond to the D'Hondt method and thus satisfies ALQ but fail AUQ.
Perpetual Consensus, when restricted to the apportionment setting, corresponds to 
Frege's apportionment method and thus satisfies AUQ but fails ALQ.
}
\begin{proposition}\propapportionment\label{prop:apportionment}
\end{proposition}

\begin{proof}[Proof sketch.]
The first part of Proposition~\ref{prop:apportionment} follows directly from results in the approval-based multi-winner literature (e.g., from \cite{Brill2017MultiwinnerApprovalRules}), in particular from the fact that Sequential PAV and \phragmen's sequential rule behave like the D'Hondt method in
the apportionment setting.
 
Frege's apportionment method~\cite{frege} is defined as follows:
Let $(p_1,\dots,p_m)$ be a tuple of reals such that $\sum_{i=1}^m=1$.
These numbers represent an apportionment instance with $m$ parties, in which party~$i$ has received a $p_i$-fraction of votes.
Frege's apportionment method is defined in rounds; in each round~$t$ a seat is assigned to the party~$i$ with highest weight $s_i^t$.
In round~1, the weights are defined as $s_i^1=p_i$ for $i\in\{1,\dots,m\}$.
The weights in round $t+1$ are defined as
\[	s_i^{t+1}		=	
	\begin{cases}
		\rule[-1em]{0pt}{2em}
		s_i^t+p_i-1	&\text{if party $i$ receives a seat,}\\
		s_i^t+p_i																&\text{otherwise.}
	\end{cases}
\]
To see that Frege's apportionment method coincides with Perpetual Consensus on apportionment instances,
let $\calD$ be a simple $k$-decision sequence.
Further, let $C=\{c_1,\dots,c_m\}$ be the set of alternatives
and for $c\in C$, $N(c)=\{v\in N:A(v)=\{c\}\}$.
We define $p_i=\frac{|N(c_i)|}{n}$ for  $i\in\{1,\dots,m\}$.
Then $\calD$ corresponds to the apportionment instance $(p_1,\dots,p_m)$.
Let $\alpha_t(v)$ denote the weight of voter~$v$ in round~$t$ for Perpetual Consensus.
For $v\in N$ with $A(v)=\{c_i\}$, we claim that 
$
|N(c_i)|\cdot \alpha_t(v)=n\cdot s_t^i.
$
(The proof for this claim is in the appendix.)
As $|N(c_i)|\cdot \alpha_t(v)$ equals $s_t^i$ times $n$ and both methods choose the alternative with the higher score, they select the same alternative (assuming the same tie-breaking is used).
Finally, Frege's apportionment method satisfies AUQ but fails ALQ \cite{frege}.
\end{proof}

\newcommand{\ALQfail}{
Perpetual Consensus 
fails ALQ. Perpetual PAV and Perpetual \phragmen{} fail AUQ.} 

\newcommand{\PQAUQ}{Perpetual Quota (new) satisfies Apportionment Upper Quota.}

\newcommand{\PCAUQ}{Perpetual Consensus satisfies Apportionment Upper Quota. Perpetual PAV and Perpetual \phragmen{} satisfy ALQ.}

\subsection{A Discussion of Stronger Proportionality Axioms}

As we have seen, already proportionality axioms from the apportionment setting yield a clear distinction between perpetual voting rules.
Let us briefly discuss two stronger axioms from the approval-based multi-winner literature: Proportional Justified Representation \cite{Sanchez-Fernandez2017Proportional} and Extended Justified Representation (EJR) \cite{Aziz2017Justifiedrepresentationin}.
We are going to argue that these two axioms, while being well-established in the multi-winner setting, 
are not particularly helpful in our analysis.

First, we recall that both axioms imply Apportionment Lower Quota (this is a simply consequence of the corresponding definitions) and thus Perpetual Consensus satisfies neither. 
Furthermore, since Perpetual PAV is identical to Sequential PAV in the approval-based multi-winner setting and Sequential PAV fails both PJR and EJR~\cite{Aziz2017Justifiedrepresentationin}, we can conclude that also Perpetual PAV fails both axioms.
Finally, Perpetual \phragmen{} is equivalent to \phragmen's Sequential Rule in the approval-based multi-winner setting and 
since \phragmen's Sequential Rule satisfies PJR and fails EJR \cite{Brill2017Phragmens}, so does Perpetual \phragmen.
We can conclude that Perpetual \phragmen{} is the only perpetual voting rule in this paper that satisfies PJR and no perpetual voting rule is known that satisfies EJR\footnote{Rule~X, as proposed by \citet{peters2020proportionality}, satisfies EJR but is not committee monotone and thus not directly adaptable to a perpetual rule. More generally, it is an open problem whether committee monotone rules exist that satisfy EJR.}.
However, it would be wrong to conclude from this argument that Perpetual \phragmen{} is the most 
proportional. First, Perpetual Consensus is incompatible with PJR just because of different behaviour in the apportionment setting (upper quota vs.\ lower quota). Second, also for Sequential PAV good arguments can be made that it is more proportional than \phragmen's rule (e.g., sequential PAV has 
a larger proportionality degree than \phragmen's rule, at least for committees of size~$\leq 100$ \cite{skowron2018proportionality}.)

Which stronger proportionality axioms are suitable for the perpetual setting?
A first step towards a more neutral comparison of proportionality could be a concept called \emph{laminar proportionality} \cite{peters2020proportionality}.
This concept is agnostic to whether a rule satisfies lower or upper quota (or neither),
as it is concerned with instances where all quotas are integral.
It is, however, in conflict with a ``welfarist'' understanding of proportionality (also defined in \cite{peters2020proportionality}), which is the kind of proportionality that (Sequential) PAV is based on.

A more flexible and fine-grained measure is the proportionality degree \cite{skowron2018proportionality} (also known as average satisfaction of cohesive groups~\cite{Sanchez-Fernandez2017Proportional}).
This notion could be adapted to the perpetual setting as follows:

\begin{definition}
Given a $k$-decision sequence $\calD=(N,\bar A, \bar C)$, a group of voters $N'\subseteq N$ is $\ell$-large if (i) for all rounds $i\in\{1,\dots,k\}$ and $u,v\in N'$, $A_i(u)=A_i(v)$ (they have the same preferences in all rounds), and (ii) $|N'|\geq \ell\cdot \frac{|N|}{k}$.

Fix a function $f\colon \mathbb{N} \to \mathbb{R}$. A perpetual rule $\calR$ has a \emph{perpetual proportionality degree} of $f$ if for each $k$-decision sequence $\calD=(N,\bar A, \bar C)$ and each each $\ell$-large group $N'$, 
\begin{align*}
\sum_{v\in N'} \frac{\satisfaction(v, \calR(\calD))}{|N'|}\geq f(\ell) \text{.}
\end{align*}
\end{definition}\label{def:propdegree}

Note that Definition~\ref{def:propdegree} is essentially a lower-quota concept but does not require exact adherence to ALQ. We further note that it follows from Proposition~2 in \cite{aaai/AEHLSS-polyejr} that $f(\ell)$ can be no larger than $\ell-1$; this would constitute an optimal guarantee.

\begin{proposition}
Perpetual \phragmen{} has a perpetual proportionality degree of $\frac{\ell-1}{2}$.
\end{proposition}
\begin{proof}
This follows from the proof that \phragmen's rule (in the approval-based multi-winner setting) has a proportionality degree of $\frac{\ell-1}{2}$ \cite{skowron2018proportionality}; essentially the same proof applies here.
\end{proof}

It remains an open problem to determine the perpetual proportionality degrees of other rules. 
For Perpetual Consensus, this will require a better understanding of Frege's apportionment method (in particular, to know by how much it may violate ALQ).
A summary of our axiomatic analysis can be found in Table~\ref{tab:overview}.

\begin{table}
\begin{center}
\begin{tabular}{lccccc}
               & Simple    & Indep.\ of    & Bounded           & \multicolumn{2}{c}{Apportionment}\\ 
                     &  Proport. & Uncontr.\ Dec. &  Dry Spells & Lower Qu.    & Upper Qu.    \\ 
\midrule                   
AV                   & \xmark    & \cmark        & \inftyno    & \xmark         & \xmark         \\ 
Per.~PAV             & \cmark    & \xmark        & \inftyno    & \cmark         & \xmark         \\ 
Per.~Reset           & \xmark    & \xmark        & \cmark      & \xmark         & \xmark         \\ 
\newrule             & \xmark    & \cmark        & \cmark      & \xmark         & \xmark         \\
Per.~Consensus       & \cmark    & \cmark        & \cmark      & \xmark         & \cmark         \\  
Rotating Dict.       & \cmark    & \xmark        & \cmark      & \xmark         & \xmark         \\
Per.~\phragmen       & \cmark    & \xmark        & \cmark      & \cmark         & \xmark         
\end{tabular}
\caption{Axiomatic results for selected perpetual voting rules.}
\label{tab:overview}
\end{center}
\end{table}

\section{Discussion}
\label{sec:discussion}

The first question tackled by this paper is whether there are perpetual rules that are conceptually simple and
satisfy all three basic axioms from \cite{Lackner2020}. 
This question can be answered mostly in the negative.
In particular, it is notable that win-based WAMs, which include 
natural adaptions of many popular multiwinner rules cannot satisfy
more than one basic axiom at once and cannot satisfy the crucial property
of limited dry-spells at all.

The second question relates to stronger proportionality 
notions. We have seen that already in the apportionment setting 
we can distinguish between, for example, Perpetual Consensus
and Perpetual \phragmen. Moreover, we observe that the slightly stronger 
proportionality notions defined in the apportionment setting
allow us to show that some rules that satisfy simple proportionality
for technical reasons, like Rotating Dictatorship, are in fact not 
proportional. Finally, we highlighted some further directions
in which the analysis of proportionality in perpetual voting can be extended.

Our work opens several interesting opportunities for future work.
First of all, a more in-depth analysis of proportionality in the perpetual 
setting is needed. Moreover, it would be interesting to search for
voting rules with better dry-spell guarantees than, for example,
Perpetual Reset. However, we have to observe that minimizing the dry-spell
guarantee is not a sufficient goal by itself.
For example, Rotating Dictatorship offers an optimal dry-spell guarantee of $n$,
but is still an unacceptable rule.
Indeed, it can be shown that any voting rule with this optimal dry-spell
guarantee must have, in the worst case, arbitrarily many rounds in
which only the opinion of one voter is considered (see Appendix~\ref{app:dict} for more details).
Finally, many other crucial axiomatic properties still have to be explored 
in perpetual voting, including, e.g., strategy-proofness.

\bibliographystyle{abbrvnat}
\bibliography{literature}

\begin{contact}
Martin Lackner\\
TU Wien\\
Vienna, Austria\\
\email{lackner@dbai.tuwien.ac.at}
\end{contact}

\begin{contact}
Jan Maly\\
TU Wien\\
Vienna, Austria\\
\email{jmaly@dbai.tuwien.ac.at}
\end{contact}

\newpage

\appendix

\section{Proof Details from Section~\ref{sec:basicWAMs}}\label{App:proofs}

First, observe that basic WAMs have the following very useful property:

\newcommand{\LemCons}{Let $\calR$ be a basic WAM and let
$(N, \bar A, \bar C)$ and $(N^*, \bar A^*,\bar C^*)$ be two $k$-decision sequences 
such that and $w_i = w_i^*$ for all $i$ under $\calR$.
Then for all voters $v \in N$ the weight of $v$ after round $k$
is the same for the decision sequences 
$(N, \bar A, \bar C)$ and $(N + N^*, \bar A + \bar A^*,\bar C +\bar C^*)$.
Here $(N + N^*, \bar A + \bar A^*,\bar C +\bar C^*)$ denotes 
the decision sequence with alternatives $C \cup C^*$,
$N + N^*$ voters where the first $N$ voters vote like the voters in
$(N, \bar A, \bar C)$
and the voters from $N+1$ to $N + N^*$ vote like the voters in
$(N^*, \bar A^*,\bar C^*)$.
}

\begin{lemma}\label{Lem:Consistency}
\LemCons
\end{lemma}

\begin{proof}
We prove the lemma by induction over $k$. 
By definition the weight of all voters in the first round is $1$.
Now assume that the weights of all voters are the same 
in $(N,\bar A,\bar C)$ resp.\ $(N^*, \bar A^*, \bar C^*)$ 
and $(N + N^*, \bar A + \bar A^*,\bar C +\bar C^*)$
in round $k-1$. Now let $w$ be the winner in round $k-1$
for $(N,\bar A,\bar C)$ and $(N^*, \bar A^*, \bar C^*)$.
Then, for every alternative $c \in C_{k-1} \cup C_{k+1}^*$ we have 
$\score_{k-1}(w) > \score_{k-1}(c)$ in $(N,\bar A,\bar C)$
and $(N^*, \bar A^*, \bar C^*)$ (if $c \not \in C_{k-1}$,
then we set $\score_{k-1}(c) = 0$ in $(N,\bar A,\bar C)$ and the same for 
$(N^*, \bar A^*, \bar C^*)$.)
As the weigth of all voters is the same in
$(N,\bar A,\bar C)$ resp.\ $(N^*, \bar A^*, \bar C^*)$ 
and $(N + N^*, \bar A + \bar A^*,\bar C +\bar C^*)$
in round $k-1$, the score of all alternatives in 
$(N + N^*, \bar A + \bar A^*,\bar C +\bar C^*)$
is just the sum of their scores in
$(N,\bar A,\bar C)$ and $(N^*, \bar A^*, \bar C^*)$.
Therefore, we have 
$\score_{k-1}(w) > \score_{k-1}(c)$ in
$(N + N^*, \bar A + \bar A^*,\bar C +\bar C^*)$
for all $c \in C_{k-1} \cup C_{k-1}^*$,
hence $w$ is the winner in round $k-1$.
However, then the weight of any voter $v$ in 
round $k$ is 
\[\alpha_{k}(v) = \begin{cases} f(\alpha_{k-1}(v)) & \text{if }w \notin A_{k-1}(v),\\
g(\alpha_{k-1}(v)) & \text{if }w \in A_{k-1}(v).\end{cases}\]
where $\alpha_{k-1}$, $A_{k-1}(v)$ and $w$ are the same for 
$(N,\bar A,\bar C)$ and $(N + N^*, \bar A + \bar A^*,\bar C +\bar C^*)$.
Hence, $\alpha_k(v)$ is the same in
$(N,\bar A,\bar C)$ and $(N + N^*, \bar A + \bar A^*,\bar C +\bar C^*)$
for all $v \in N$. 
\end{proof}

\begin{repproposition}{prop:notconst}
\notconst
\end{repproposition}
\begin{proof}
Let $\calR$ be either a win-based or a loss-based WAM, i.e., either $g(x) = x$ or $f(x) =x$ holds.
We assume for the sake of a contradiction
that there is a function $b$ that bounds the dry spells of $\calR$.

First, assume $g(x) = x$ and there is a value $x^*$ (that can appear) such that $f(x^*) = x^*$.
[$f(x) = x$ and there is a value $x^*$ (that can appear) such that $g(x^*) = x^*$].
Then, let $(N, \bar A, \bar C)$ be a $k$-decision sequence such that 
there is a voter that has weight $x^*$ in round $k+1$.
We consider the decision sequence that consist of three copies of 
$(N, \bar A, \bar C)$. Then, by Lemma~\ref{Lem:Consistency}
there are three voters $v_1, v_2, v_3$ with weight $x^*$
in round $k+1$. Now, consider that in the next $b(3|N|)+1$ rounds
there are two alternatives $\{a,b\}$,
voter $v_1$ votes for a alternative $a$, voter $v_2$ and $v_3$ vote
$b$ and every other voter votes $\{a,b\}$.
Then the score of alternative $a$ is always $x^* + V$,
where $V$ is the weight of the voters in $V \setminus \{v_1,v_2,v_3\}$,
and the score of alternative $b$ is always $2x^* +V$. 
Therefore, $b$ has a higher score in every election
and $v_1$ has a dry spell of $b(3|N|)+1$ rounds. A contradiction.

Now, assume that $g(x) = x$ and $f(x) > x$ 
[or $f(x) = x$ and $g(x) < x$] on all values $x$ that can appear as 
weights of voters.
Now consider an election with two voters $v_1$ and $v_2$, two alternatives $a$ and $b$
and $2b(3) +3$ rounds, where $v_1$ always votes $a$ and $v_2$ always votes $b$.
Then, either $v_1$ or $v_2$ looses at least $b(3) +2$ rounds.
We assume w.l.o.g.\ that $v_1$ looses $b(3) +2$ rounds.
Then, the weight of voter $v_1$ in round $2b(3) +4$ is
$f^{b(3) +2} (1)$ [$g^{b(3)+1} (1)$]. Now, we add another voter $v_3$ that 
votes $\{a,b\}$ in the first $2b(3) +3$ rounds. Hence, he wins in every 
round and his weight in round $2b(3) +4$ is $1$ [$g^{2b(3) +3}(1)$].
Then, assume that in the next $b(3) +1$ rounds voters $v_1$ and $v_2$
vote $a$ and voter $v_3$ votes $b$. We claim that voter $v_3$ wins non of these
$b(3) +1$ rounds:
Alternative~$a$ always has a score of at least $f^{b(3) +2} (1)$ [$g^{2b(3)+2} (1)$],
whereas alternative $b$ has a score of at most $f^{b(3) +1} (1)$ [$g^{2b(3)+3} (1)$].
As $f(x) > x$  [$g(x) < x$], alternative $a$ has a higher score in every 
election.
Therefore, $v_3$ has a dry spell of $b(3) +1$, a contradiction.
\end{proof}

\begin{repproposition}{prop:IUDident}
\IUDident
\end{repproposition}

\begin{proof}
First observe that $f(x) \geq x \geq g(x)$ for all $x$,
hence $f(x) \geq g(x)$.
Proof by contradiction.
First assume that $k$ is the minimal number of rounds 
such that a weight $x$ can occur for which $g(x) = 0$ holds.
Furthermore, let $(N, \bar A, \bar C)$ be a $k$ decision sequence,
such that the weight of voter $v_i$ after round $k$ is $x$.
Then consider the case that $C_{k+1} = \{a,b\}$,
$A_{k+1}(v_j) = \{a,b\}$ for all $j \neq i$ and $A_{k+1}(v_i) = \{b\}$.
Then $w_{k+1} = b$ as $x \neq 0$ by assumption.
However, if we consider the case that $C_{k+1}= C_{k+2} = \{a,b\}$,
$A_{k+1}(v_j) = \{a,b\}$ for all $j$,
$A_{k+2}(v_j) = \{a,b\}$ for all $j \neq i$ and $A_{k+2} = \{b\}$,
then $a$ and $b$ have the same score in round $k+2$ because $g(x) = 0$.
Therefore, we can assume w.l.o.g.\ that by tie-breaking $w_{k+2} = a$.
This contradicts independence of uncontroversial decisions.

Now, assume that there is a $k$ such that there are weights $x$ and $y$ 
that can occur after $k$ rounds for which $g(x) = cx$ and $g(y) = cy+d$
for some $c,d > 0$.
Furthermore, let $(N, \bar A, \bar C)$ be a $k$ decision sequence,
such that the weight of voter $v_i$ is $x$
and let $(N^*, \bar A^*, \bar C^*)$ be a $k$ decision sequence,
such that the weight of voter $v^*_i$ is $y$.
We can assume that $w_i = w_i^*$ for all $i \leq k$ by renaming the alternatives
if necessary. Then, by Lemma~\ref{Lem:Consistency},
$(N + N^*,\bar A + \bar A^*,\bar C + \bar C^*)$
is a $k$ decision sequence such that the weight of voter $v_i$
is $x$ and the weight of voter $v_i^*$ is $y$.

Next, we compute integers $m$, $m'$ such that 
$mx > m'y$ and $mx - m'y < m'd$.
In other words, we compute values such that
$0 < mx - m'y < m'd$. This is equivalent to
\[\frac{m'y}{x} < m < \frac{m'y}{x} + \frac{m'd}{x}.\]
Clearly, for large enough integers $m'$, there is an 
integer $m$ that satisfies this equation.
Let $m^* = \max(m,m')$. 
Then, we consider the decision sequence 
\[(m^*\cdot(N + N^*),m^*\cdot(\bar A + \bar A^*),m^*\cdot(\bar C + \bar C^*)).\]
By Lemma~\ref{Lem:Consistency} there are $m^*$ voters with weight $x$
and also $m^*$ voters with weight $y$ after round $k$.
Then consider the case $C_{k+1} = \{a,b\}$,
$A_{k+1}(v) = a$ for $m$ voters with weight $x$,
$A_{k+1}(v) = b$ for $m'$ voters with weight $y$
and $A_{k+1}(v) = \{a,b\}$ for everyone else.
Then $w_{k+1} = a$ as the score of $a$ in round $k+1$
is $mx +C$ for some $C$ and the score of $b$ is $m'y +C$.
However, if we consider the case that $C_{k+1}= C_{k+2} = \{a,b\}$,
$A_{k+1}(v_j) = \{a,b\}$ for all $j$,
$A_{k+2}(v) = a$ for $m$ voters with weight $x$,
$A_{k+2}(v) = b$ for $m'$ voters with weight $y$
and $A_{k+2}(v) = \{a,b\}$ for every one else, then
$w_{k+2} = b$ as the score of $a$ in round $k+2$
is $cmx +C$ for some $C$ and the score of $b$ is $cm'y + cm'd +C$.
This contradicts independence of uncontroversial decisions.
\end{proof}

Let us recall the definition of \newrule.

\newruleDef

\begin{repproposition}{prop:newruleBD}
\newruleBD
\end{repproposition}

\begin{proof}
We claim that using the \newrule{} any weight $x$ such that $r(x) = k$
can only appear in round $k$.
Indeed, in round $k$ only weights of the following form
(for some integer $l \geq 0$) can appear:
\[x = \frac{2(k-1)+1}{2^{l}}\]
which satisfies $r(x) = k$.
We can show this by induction on $k$. For $k =1$ only weight $1$
can appear. This equals 
\[\frac{2(1-1)+1}{2^{0}}.\]
Now, assume the claim holds in round $k$. This means. 
all weights in round $k+1$ are of the form
\[\left(\frac{2r(x)+1}{2^{l'}(2r(x)-1)}\right)x.\]
where $x$ is a weight that can appear in round $k$
and $l' = 0$ if $f$ is applied to $x$.
By the induction hypothesis $r(x) = k$, hence we have
\[\left(\frac{2k+1}{2^{l'}(2k-1)}\right) x = \left(\frac{2k+1}{(2k-1)2^{l'}}\right)
\left(\frac{2(k-1)+1}{2^{l}}\right) = \frac{2k+1}{2^{l'+l}}.\]
This concludes the proof of the claim.
It follows from the claim that, for any value $x$ that can appear in round $k$
we have 
\[g(x) = \frac{2k+1}{2^{k!}} x,\]
which is multiplication by a constant for fixed $k$.
Therefore, \newrule{} satisfies independence of uncontroversial decisions
by Proposition~\ref{prop:IUDident}. 

It remains to show that \newrule{} has bounded dry spells.
We claim that the dry spells are bounded by $b(n) = n + h(n)$,
where $h$ is a function such that $2^{k!} > n$ for all $k > h(n) -1$.
In the following, we say a voter $v$ has a $h(n)$-dry spell in round $k$ of length $l$
if $v$ did not approve the winner in the rounds $k-1, k -2, \dots, k-l$ and
$k-l \geq h(n)$, i.e., if $v$ has a dry spell of length $l$ ignoring rounds before $h(n)$.
We claim that in every round $l \geq h(n)$ it holds that at least one voter 
that has a maximal $h(n)$-dry spell supports the winning alternative.
If all voter have maximal $h(n)$-dry spell, then the claim holds trivially.
Therefore, assume that there is at least one voter with non-maximal $h(n)$-dry spell.
Let $k+1$ be the last round in which voters with the second longest $h(n)$-dry spell
in round $l$ approved the winner. By construction $k+1 \geq h(n)$.
First we observe that the voting power of a voter $v$ with maximal $h(n)$-dry spell
in round $l$ is at least
\[\frac{2l-1}{\Pi_{i=1}^k 2^{i!}} =
\frac{2l-1}{2^{\sum_{i=1}^k i!}} \geq \frac{2l-1}{2^{k(k!)}}\]
Now, consider a voter $v'$ with a shorter $h(n)$-dry spell. The voting power of 
$v'$ in round $l$ is at most
\[\frac{2l-1}{2^{(k+1)!}}\]
Therefore, the support of a alternative that is not supported by a 
single voter with maximal $h(n)$-dry spell is at most
\[(n-1)\frac{2l-1}{2^{(k+1)!}} < 2^{k!} \frac{2l-1}{2^{(k+1)!}} = 
2^{k!} \frac{2l-1}{2^{(k+1)(k!)}}= 2^{k!} \frac{2l-1}{2^{k(k!)}2^{k!}} =
\frac{2l-1}{2^{k(k!)}}\]
hence smaller than the weight of a single voter with maximal $h(n)$-dry spell.
Therefore, only alternatives supported by at least one voter with maximal $h(n)$-dry spell can win.

It follows from this claim that after round $h(n)$ no dry spells longer than $n$ can occur:
Assume a voter $v$ has a dry spell $n$ in round $k \geq h(n) +n $.
We claim that the winner in round $k$ must be supported by $v$.
In every round $k -i$ for $i \leq n$ by our claim there must have
been a voter $v'$ that had a $h(x)$-dry spell at least as long as $v$ among the supporters of the
winning alternative. This means this voter did not win in round $k-n, \dots, k - i -1$. 
As there are only $n-1$ voters apart from $v$ each one of them must have won at least one
round $k -i$ for $i \leq n$. Therefore, $v$ is the only voter with maximal $h(n)$-dry spell
in round $k$, which implies that the winner in round $k$ must be supported by $v$.
\end{proof}

\begin{proposition}\label{newruleFailSP}
\newrule{} does not satisfy simple proportionality.
\end{proposition}

\begin{proof}
Consider $(\{a\},\{a\},\{a\},\{a\},\{b\})^5$.
The corresponding choice sequence is $(a,a,b,a,b)$, which is not proportional.
\end{proof}

\begin{repproposition}{simplepropChar}
\simplepropChar
\end{repproposition}

\begin{proof}
Assume that $xw_x<(y+1)w_y$ for all integers $x,y\geq 0$.
If we set $y=x-1$, we obtain 
\begin{align}
w_x<w_{x-1}\label{eq:strict}
\end{align}
Towards a contradiction, assume that the rule fails simple proportionality for some simple $|N|$-decision sequence $(N, \bar A, \bar C)$ and corresponding $|N|$-choice sequence $\bar w$.
Now, let $k$ be the first round such that there is a voter $v$ with
$\satisfaction(v, \bar w_{\leq k+1})=\#v+1$. Such a round must exist,
because simple proportionality is violated.
In round $k+1$, there also exists a voter $v'$ with $\satisfaction(v', \bar w_{\leq k+1})<\#v'$.
Let $\#v=x$ and $\#v'=y$.
Observe that, by assumption, $\satisfaction(v,\bar w_{\leq k}) = \#v$.
As $v$ wins in round $k+1$, it holds that 
$xw_{\satisfaction(v,\bar w_{\leq k})} = xw_x \geq y\cdot w_{\satisfaction(v', \bar w_{\leq k})}$.
Since $\satisfaction(v', \bar w_{\leq k})=\satisfaction(v', \bar w_{\leq k+1})<\#v'=y$, by \eqref{eq:strict} we have $w_{\satisfaction(v', \bar w_{\leq k})}\geq w_{y-1}$.
Thus, $xw_x \geq y\cdot w_{\satisfaction(v', \bar w_{\leq k})}\geq y w_{y-1}$, a contradiction.

For the other direction, assume that simple proportionality holds.
First, let us show that $w_{x}<w_{x-1}$.
Consider a simple $(2x)$-simple profile with  $x$ voters approving some alternative and the remaining $x$ voters approving some other alternative. In round $2x$, one of these groups has satisfaction $x-1$, the other $x$. By simple proportionality, the group with satisfaction $x-1$ must win (independent of tiebreaking).
Thus, $xw_x < xw_{x-1}$, i.e., $w_{x}<w_{x-1}$.

Next, we show that $xw_x<(y+1)w_y$.
Consider an $(x+y+1)$-simple profile with two groups: $x$ voters approving one alternative, 
$y+1$ voters approve another alternative.
Consider the round $k$ where the latter group wins for the $(y+1)$-st time (this has to happen due to simple proportionality).
Let $x'\leq x$ be the satisfaction of the former group in round $k$.
Assuming that tiebreaking is against this group, it holds that
$(y+1)w_y > x w_{x'} \geq x w_x$.
\end{proof}

\begin{repproposition}{prop:winbasedspropIUD}
\winbasedspropIUD
\end{repproposition}

\begin{proof}
We observe that for a win-based WAM the weight $1$ can occure in any round $k$. 
Indeed, consider a $k$-decision sequence with $k$ voters $v_1, \dots, v_{k}$ and $k$
alternatives $a_1, \dots, a_{k}$ such that in each round voter $v_i$ approves only
$a_i$. By construction, there is at least one voter $v_j$ in round $k$ who did not appove of 
any winner. Hence, $v_j$ has weight $1$ in round $k$. 

It follows that a win-based WAM satisfies independence of uncontroversial decisions
if and only if there is a $c$ such that $g(x) = cx$ for all $x$ that can appear as 
weights in any round.
We claim that such a rule can not satisfy simple proportionality. 
By Proposition~\ref{simplepropChar}, we know $2c = 2 w_2 < 2 \cdot w_1 = 2$,
hence $c < 1$ must hold.
Furthermore, for every $y \geq 2$ we have the following:
\begin{equation*}
 1 \cdot w_1 < ((y-1)+1)w_{y-1} \Leftrightarrow
 1 < yc^{y-2} \Leftrightarrow
 \frac{1}{y} < c^{y-2} \Leftrightarrow
 \left(\frac{1}{y}\right)^{\frac{1}{y-2}} < c
\end{equation*}
However, as $\left(\frac{1}{y}\right)^{\frac{1}{y-2}}$ gets arbitrarily close to $1$ for large enough $y$,
this contradicts $c <1$.
\end{proof}

\section{Proof Details from Section~\ref{sec:proportionality}}

\begin{proposition}\label{RotDictSP}
Rotating Dictator satisfies simple proportionality.
\end{proposition}

\begin{proof}
Let $\calD$ be a simple $|N|$-decision sequence. Further, let $N(c)=\{v\in N:A(v)=\{c\}\}$ for $c \in C$.
Then, $c$ is picked as the winner whenever a voter from $N(c)$ is picked as dictator.
Hence for every candidate $c$ and voter $v \in N(c)$ we have
$\satisfaction(v, \bar w) = |N(c)| = \#v$.
\end{proof}

\begin{repproposition}{prop:phragmen-props}
\phragmenprops
\end{repproposition}
\begin{proof}
\begin{enumerate}[label=(\roman*)]
\item Simple proportionality follows immediately from the apportionment upper quota axiom, which Perpetual \phragmen{} satisfies by Proposition~\ref{prop:apportionment}.
\item First, observe that the minimum load and maximum load of voters cannot differ by more than 1; otherwise a different choice (in favour of those with a smaller load) would have been made previously.
Further, each round the total load increases by 1.
Towards a contradiction, assume that a voter~$v$ is not satisfied by a sequence of $2n-1$ decisions.
Let $\alpha$ be the load of voter~$v$ during these $2n-1$ rounds.
Further, let $N'=N\setminus \{v\}$.
The total load the voters in $N'$ is at least $(n-1)(\alpha-1)$ before these rounds, and at least 
$(n-1)(\alpha-1)+2n-1$ after these $2n-1$ rounds.
Thus, the average load of voters in $N'$ is $\geq \alpha + 1 + \frac{1}{n-1}$.
Hence, there is one voter with a load strictly larger than $\alpha + 1$, a contradiction.
Hence, we obtain a dry spell guarantee of $2n-1$.

To see that this bound is tight, consider the following decision sequence with $n$ voters.
All voters have disjoint approval sets. In round $1$, tiebreaking in favour of voter~$n$'s alternative, in later rounds it is always against voter $n$. Voter $n$ loses in rounds $2,\dots, 2n-1$, but wins in round~$2n$. This is a dry spell of length~$2n-2$.
\item Consider $(\{a\},\{b\})^2$. The corresponding decision sequence is $(a,b)$.
If we introduce $(\{a\},\{a\})$ after the first round, the corresponding decision sequence is $(a,a,a)$.
This is because before round 2 voter 1's load is 1 and of voter 2's load is 0. Thus, in round 2, the load of voter 2 increases to 1 and voter 1's load remains 1. Now, in round 3, there is a tie between alternative~$a$ and $b$ and by alphabetic tiebreaking $a$ wins.
\item 
We consider a decision sequence with 7 voters.
In the first two rounds, the preferences are \[(\{a\},\{a\},\{a\},\{b\},\{b\},\{b\},\{c\}).\]
Thus, $a$ wins in the first round (we assume alphabetic tiebreaking) and the corresponding loads are $(\frac 1 3,\frac 1 3,\frac 1 3,0,0,0,0)$.
In the second round $b$ wins and the loads are $(\frac 1 3,\dots,\frac 1 3,0)$.
Assume towards a contradiction that Perpetual \phragmen{} is a WAM. 
We can thus assign weights in some fashion; let these be $x_1,\dots,x_7$.

We now consider several decision instance for round three.
First, if the preferences are \[(\{a\},\{b\},\{c\},\{d\},\{e\},\{f\},\{a,b,c,d,e,f\}),\] then all alternatives are tied.
Thus, we can conclude that $x_1=x_2=x_3=x_4=x_5=x_6$; let $x=x_1=\dots=x_6$ and $y=x_7$.
Second, if the preferences are \[(\{a\},\{a\},\{b\},\{c\},\{d\},\{e\},\{f\}),\] then $a$ wins (the load of voter 1 and 2 would increase to $\frac 5 6$).
Thus, we infer that $2x > y$.
Finally, we consider \[(\{a\},\{a\},\{a\},\{b\},\{c\},\{c\}).\]
Here $a$ and $c$ are tied; in both cases the load of the corresponding voters would increase to $\frac 2 3$.
Thus, it holds that $3x = x+y$ and in turn $2x = y$. This contradicts our previous result that $2x > y$.
We conclude that Perpetual \phragmen{} cannot be ``simulated'' by a WAM.

\item To calculate $\ell_{k+1}(v)$ for all $v\in N$, we have to first find the set of voters $N'$ for which $\ell_{k+1}(N')$ is minimal. Recall that
\[\ell_{k+1}(N')=\frac{1 + \sum_{v\in N'} \ell_k(v)}{|N'|}.\]
Let $N(c)=\{v\in N:A(v)=\{c\}\}$ for $c \in C$. 
We will calculate for each alternative $c\in C_{k+1}$ the subset $N'(c)\subseteq N(c)$ that has a minimal $\ell_{k+1}(N'(c))$ among all subsets of $N(c)$.
To do this we sort voters in $N(c)$ by their load: $\ell_{k}(v_1)<\ell_{k}(v_2)<\dots<\ell_{k}(v_s)$ with $N(c)=\{v_1,\dots,v_s\}$.

We claim that $N'(c)$ is an interval containing $v_1$ in this order, i.e., there exists a $t\leq s$ such that $N'(c)=\{v_1,\dots, v_t\}$. Towards a contradiction, assume that $N'(c)$ does contain $v_t$ but not $v_r$ with $r<t$. Then clearly replacing $\ell_{k+1}(N'(c))> \ell_{k+1}(N'(c)+v_t-v_r$); a contradiction.

Since $N'(c)$ consists of an interval when alternatives are sorted by load, we can determine the optimal value by adding alternatives one by one (starting with the lowest-load alternative, i.e.,~$v_1$).
Then, we compare $N'(c)$ for all $c\in C_{k+1}$ and thus find $N'$ with minimal $\ell_{k+1}(N')$.
This procedure requires polynomial time.
\end{enumerate}
\end{proof}

\begin{proposition}\label{RotDictALQ}
Rotating Dictator satisfies neither ALQ nor AUQ.
\end{proposition}

\begin{proof}
Consider a simple $2$-decision sequence with five voters and two candidates $a$ and $b$.
Furthermore, assume that two voters approve $a$ and three voters approve $b$.
Now, additionally assume that the two voters approving $a$ are picked as dictators first,
hence $a$ is the winner in the first two rounds.
Let $v$ be a voter approving $a$. Then, in round $2$ we have 
\[\satisfaction(v,(a,a)) = 2 > \left\lceil 2 \cdot \frac{2}{5} \right\rceil = \left\lceil k \cdot \frac{\#v}{n} \right\rceil.\]
Hence AUQ is violated. Now, let $v'$ be a voter approving $b$. Then, in round $2$ we have 
\[\satisfaction(v',(a,a)) = 0 < \left\lfloor 2 \cdot \frac{3}{5} \right\rfloor = \left\lfloor k \cdot \frac{\#v'}{n} \right\rfloor.\]
Therefore, ALQ is violated.

\end{proof}

\begin{reptheorem}{thm:pav-char}
\pavchar
\end{reptheorem}
\begin{proof}
Let $w = (w_1 =1, w_2 = g(1), w_3 = g(g(1)), \dots)$. We claim that $w_j = \nicefrac 1 j$
must hold for all $j$.
First assume that there is a $j \geq 1$ such that $w_j > \nicefrac 1 j$.
Then, there is a $\epsilon > 0$ such that $w_j = \nicefrac 1 j +\epsilon$.
We pick a $k > \lceil \nicefrac{1}{j \epsilon} \rceil$ such that $j$ divides $k$.
Now, we consider the following $k$-decision sequence:
Let $t = \nicefrac k j$. Then, in every round we have $C = \{c_0, \dots , c_t\}$.
Furthermore, $N = \{1, \dots , k^2\}$ and in every round the same $k$ voters vote for $c_0$
and for every $c_i$ with $i \geq 1$ there are $j(k-1)$ unique voters that vote $c_i$ each round.
We observe that this is an apportionment-like decision sequence. Furthermore,
$k + tj(k-1) = k^2$, hence it is actually possible to pick unique voters for every alternative.
Now, the quota of alternative $c_0$ in each round is $\nicefrac k k^2$. Hence, the quota of 
$c_0$ is $1$ in round $k$. Because $\calR$ satisfies ALQ this implies that $c_0$ is the winner in at least 
one round.
All other alternatives have the same number 
of voters supporting them and as we assume that $g$ is decreasing, 
the number of wins of the other alternatives can not differ by more than 1.
As there are $t$ other alternatives, this means for all $i \geq 1$ that $c_i$
can win at most $\nicefrac k t = j$ rounds. Therefore the support
of alternative $c_i$ in every round is at least 
\[j(k-1)\cdot\left(\frac{1}{j} + \epsilon\right) = k - 1 + j\epsilon(k-1) >
k - 1 - \frac{j\epsilon}{j\epsilon} = k\]
However, the support of $c_0$ is at most $k$. Hence, $c_0$ can not win a single round.
A contradiction.

Now, assume there is a $j$ such that $w_j < \nicefrac 1 j$.
Then, there is a $\epsilon > 0$ such that $w_j = \nicefrac 1 j -\epsilon$.
We pick a $k >j + \lceil \nicefrac{1}{\epsilon} \rceil$.
Now, we consider the following $k$-decision sequence:
In every round we have $C = \{c_0, \dots, c_{k-j+1}\}$.
Furthermore, $N = \{1, \dots , k(k-j+1)\}$ and in every round the same $j(k-j+1)$ voters vote for $c_0$
and for every $c_i$ with $i \geq 1$ there are $k-j$ unique voters that vote $c_i$ each round.
We observe that this is an apportionment-like decision sequence. Furthermore,
$j(k-j+1) + (k-j)(k-j+1) = k(k-j+1)$, hence it is actually possible to pick unique voters for every alternative.
Now, the quota of alternative $c_0$ in each round is 
\[\frac{j(k-j+1)}{k(k-j+1)} = \frac{j}{k} = k-j + 1 - kj\epsilon + j^2\epsilon  - j\epsilon < 
k - j + (1-j) \leq k-j.\]
Hence, the quota of $c_0$ is $j$ in round $k$.
Because $\calR$ satisfies ALQ this implies that $c_0$ is the winner in at least 
$j$ round.
We claim that $c_0$ only wins $j-1$ rounds. 
Assume we are in a round, such that $c_0$ has already won $j-1$ rounds.
Then, the score of $c_0$ is 
\[j(k-j+1)\left(\frac{1}{j} - \epsilon \right).\] 
All other alternatives have the same number 
of voters supporting them and as we assume that $g$ is decreasing, 
the number of wins of the other alternatives can not differ by more than 1.
As there are $k-j+1$ other alternatives and $c_0$ wins $j-1$ rounds,
this means for all $i \geq 1$ that $c_i$
can win at most \[\frac{k-j+1}{k-j+1} = 1\] rounds. Therefore,
in every round, there is at least one alternative $c_i$ that has not won a single round yet.
The score of that alternative is $k-j$. 
However, this is larger than the score of $c_0$. Hence, $c_0$ can not win a $j$th round.
A contradiction.

It remains to show that Perpetual PAV satisfies ALQ. This follows again from the fact that Perpetual PAV is identical to D'Hondt method in the apportionment
setting. See the proof of Theorem~\ref{thm:satisfysp} for more details.

Let us show that the characterization does not hold if the restriction $f(x)=x$ is dropped, i.e., if arbitrary $f$-functions are allowed.
To see this, consider the basic WAM $\calR_f$ defined by:
\begin{align*}
f(x)&= \begin{cases} \pi\cdot x & \text{if $x$ is a rational number},\\
x & \text{otherwise}.\end{cases} \\
g(x)&= 
\begin{cases} \frac{1}{\satisfaction{} + 1} & \text{if $x$ is a rational number},\\
\frac{\pi}{\satisfaction{} + 1} & \text{otherwise}.\end{cases}
\end{align*}
This function essentially resembles PAV but the weight of a voter is multiplied by $\pi$ the first time she is unsatisfied with a decision. This change does not impact ALQ:
Towards a contradiction let $\calD$ be a simple $k$-decision instance for which $\calR_f$ violates ALQ.
Let $N'$ be a representative subset of voters, i.e.,
for each $v\in N$ there exists a $u'\in N'$ such that $A(v)=A(v')$,
and further let $N'$ be minimal (no proper subset is representative).
As $\calR_f$ fails ALQ on $\calD$, there exists a $u\in N'$ such that 
\[\satisfaction(u, \calR(\calD))<\left\lfloor{k\cdot \frac{\#u}{n} }\right\rfloor.\]
Now note that 
\[\sum_{v'\in N'} \satisfaction(v', \calR(\calD)) = k.\]
Hence, there exists a $w\in N'$ with
 \[\satisfaction(w, \calR(\calD))>\left\lceil k\cdot \frac{\#w}{n} \right\rceil.\]
 Now consider the round where the satisfaction of $w$ increased to more than $\left\lceil k\cdot \frac{\#w}{n} \right\rceil$. 
 At this point, the weight of $w$ is at most \[\frac{\pi}{\left\lceil k\cdot \frac{\#w}{n} \right\rceil},\]
 the weight of $u$ is strictly larger than 
 \[\frac{\pi}{\left\lfloor{k\cdot \frac{\#u}{n} }\right\rfloor}.\]
 (The weight is multiplied by $\pi$ because $v$ was unsatisfied by at least one decision where $w$ won.)
 Now observe that
 \[\#w \cdot \frac{\pi}{\left\lceil k\cdot \frac{\#w}{n} \right\rceil} \leq  \#u\cdot \frac{\pi}{\left\lfloor{k\cdot \frac{\#u}{n} }\right\rfloor}.\]
 Consequently, the total weight of voters with the same preference as $u$ have a larger weight than all voters with the same preference as $w$.
 This contradicts the assumption that $w$'s satisfaction increases this round.
\end{proof}

\begin{repproposition}{prop:apportionment}
\propapportionment
\end{repproposition}
\begin{proof}
The first part of Proposition~\ref{prop:apportionment} follows directly from results in the approval-based multi-winner literature (e.g., from \cite{Brill2017MultiwinnerApprovalRules}), in particular from the fact that Sequential PAV and \phragmen's sequential rule behave like the D'Hondt method in
the apportionment setting.
 
Frege's apportionment method~\cite{frege} is defined as follows:
Let $(p_1,\dots,p_m)$ be a tuple of reals such that $\sum_{i=1}^m=1$.
These numbers represent an apportionment instance with $m$ parties, in which party~$i$ has received a $p_i$-fraction of votes.
Frege's apportionment method is defined in rounds; each round~$t$ a seat is assigned to the party~$i$ with highest weight $s_i^t$.
In round~1, the weights are defined as $s_i^1=p_i$ for $i\in\{1,\dots,m\}$.
The weights in round $t+1$ are defined as
\[	s_i^{t+1}		=	
	\begin{cases}
		\rule[-1em]{0pt}{2em}
		s_i^t+p_i-1	&\text{if party $i$ receives a seat,}\\
		s_i^t+p_i																&\text{otherwise.}
	\end{cases}
\]
To see that Frege's apportionment method coincides with Perpetual Consensus on apportionment instances,
let $\calD$ be a simple $k$-decision sequence.
Further, let $C=\{c_1,\dots,c_m\}$ be the set of alternatives
and for $c\in C$, $N(c)=\{v\in N:A(v)=\{c\}\}$.
We define $p_i=\frac{|N(c_i)|}{n}$ for  $i\in\{1,\dots,m\}$.
Then $\calD$ corresponds to the apportionment instance $(p_1,\dots,p_m)$.
Let $\alpha_t(v)$ denote the weight of voter~$v$ in round~$t$ for Perpetual Consensus.
For $v\in N$ with $A(v)=\{c_i\}$, we claim that 
\begin{align}
|N(c_i)|\cdot \alpha_t(v)=n\cdot s_t^i.\label{eq:frege-cons}
\end{align}
For $t=1$, this holds since $|N(c_i)|\cdot \alpha_t(v)=|N(c_i)|=n\cdot p_i = n\cdot s_t^1$.
By induction, if $c_i$ won in round~$t$,
then 
\begin{align*}
|N(c_i)|\cdot \alpha_{t+1}(v)&=|N(c_i)|\cdot \alpha_{t}(v) -n + |N(c_i)|\\
&= n\cdot s_t^i -n + |N(c_i)| \\
&= n\cdot (s_t^i + \frac{|N(c_i)|}{n} - 1) \\
&= n\cdot (s_t^i + p_i - 1) \\
&=n\cdot s_{t+1}^i.
\end{align*}
If $c_i$ did not win in round~$t$,
then 
\begin{align*}
|N(c_i)|\cdot \alpha_{t+1}(v)&=|N(c_i)|\cdot \alpha_{t}(v) |N(c_i)|\\
&= n\cdot s_t^i |N(c_i)| \\
&= n\cdot (s_t^i + \frac{|N(c_i)|}{n}) \\
&= n\cdot (s_t^i + p_i ) \\
&=n\cdot s_{t+1}^i.
\end{align*}
Now note that the score of alternative $c$ in round~$t$ with Perpetual Consensus is
$\score_{k+1}(c) = |N(c_i)|\cdot \alpha_t(v)$ (this holds for every voter $v\in N(c_i)$).
The weight of alternative~$c$ in Frege's apportionment method is  $s_t^i$.
As these two number are proportional (Equation~\ref{eq:frege-cons}) and both methods choose the alternative with the higher score, they select the same alternative (assuming the same tiebreaking is used).

Finally, Frege's apportionment method satisfies AUQ but fails ALQ \cite{frege}.
\end{proof}

\section{Dry Spells and Dictators}\label{app:dict}

\begin{proposition}
Rotating Dictator has a dry spell guarantee of $n$. This guarantee is tight, i.e., dry spells of length $n-1$ may occur.
\end{proposition}

\begin{proof}
By definition.
\end{proof}

By definition, it is impossible to have a dry spell guarantee of less than $n$.
Even though Rotating Dictator has this very good dry spell guarantee, it has the very undesirable property
that in every round only the opinion of one voter is considered.
It turns out, that this can not be avoided in the worst case if we expect a dry spell guarantee of $n$.

\begin{definition}
Let $\calH = (N,\bar A, \bar C, \bar w)$ 
be a $k$-decision history.
We say a round $j\leq k$ is \emph{dictatorial} if
there exists an alternatives $c\neq w_j$ and a voter $v\in N$ such that
$w_j \in A_{j}(v)$ and for all $v' \neq v$ it holds that $w_j \notin A_{j}(v')$ and $c \in A_j(v')$.

A perpetual voting rule $\calR$
has \emph{bounded dictatorial rounds} if there exists a function $d$ from $\mathbb{N}$ to $\mathbb{N}$ such that for any decision sequence $\calD=(N,\bar A, \bar C)$
and $\bar w = \calR(\calD)$ there are fewer than $d(|N|)$ consecutive dictatorial rounds.
\end{definition}

\newcommand{\strgdictrounds}{Every voting rule with a dry spell guarantee of $d(n)=n$ has unbounded dictatorial rounds.}
\begin{proposition}
\strgdictrounds\label{prop:strgdictrounds}
\end{proposition}

\begin{proof}
For the sake of a contradiction, assume that $\calR$ is a voting rule
with bounded dry spell guarantee $n$ and bounded dictatorial rounds.
Furthermore, let $d$ be the function bounding the number of dictatorial rounds.
Consider the following $k$-decision sequence, with $k = n + d(n)$:
$C_i = \{c_1, \dots ,c_n\}$ for all $i \leq k$.
$A_i(v_j) = c_j$ for every $i \leq n$.

We claim that every voter wins exactly once in the first $n$ rounds.
Assume otherwise that there is a voter $v^*$
that wins twice in the first $n$ rounds.
By definition, in each round at most one voter can win. 
Hence there must be a voter who does not win in the 
first $n$ rounds. However, this contradicts the assumption
that $\calR$ has a bounded dry spell guarantee of $n$.
We assume w.l.o.g.\ that voter $v_i$ wins in round $i$.

Now for the remaining rounds, we set
\[
A_i(v_j) =
\begin{cases} 
c_1 &\text{ if } i = j \mod n \\
c_2 &\text{ else.}
\end{cases}
\]

We claim that for every $i$  
we have $w_i = A_i(v_{j})$ for $j = i \mod n$ and $w_i \not \in A_i(v_{j^*})$
for all $j^* \neq i\mod n$.
We show this by induction.
By the argument above, this holds for the first $n$ rounds. 
Now assume that the claim holds for some $i \geq n$.
Then, in round $i +1$ the dry spell of voter $j$ with $j = i+1 \mod n$ is 
$n -1$ as $i+1 \neq i - l \mod n$ for all $0 \leq l < n-1$.
Hence, the bounded dry spell guarantee of $n$ 
implies $w_{i+1} \in  A_{i+1}(v_{j})$.
By the construction of $\bar A$ this implies
the claim.
This shows that there are $d(n)$ consecutive dictatorial rounds.
This is contradiction! 
\end{proof}

Voting rules with less strict dry spell bounds can avoid this very undesirable behavior,
but every rule that has limited dry spells must have some dictatorial rounds.

\newcommand{\weakdictrounds}{Every voting rule with bounded dry spells has dictatorial rounds.}
\begin{proposition}
\weakdictrounds\label{prop:weakdictrounds}
\end{proposition}

\begin{proof}
Let $\calR$ be a voting rule with limited dry spells. Let $d$ be the minimal function 
bounding the dry spells of $\calR$. Then by the minimality of $d$, there is a $k$-decision sequence 
$(N,\bar A, \bar C)$ where some voter $v$ has a dry spell of $d(N)-1$ in round $k$.
Then, voter $v$ must win in round $k+1$, hence round $k+1$ is dictatorial.
\end{proof}

\section{Proofs from the Original Perpetual Voting Paper}\label{App:orig-proofs}

The following proofs are from the original paper on perpetual voting \cite{Lackner2020} and are included here for completeness.

In the following we use a special notation for defining decision sequences.
As usual, we write $(A(1), A(2),\dots, A(n))$ to denote an approval profile.
For such an approval profile $A$, we write $A^k$ to denote the $k$-decision sequence 
where in each round approval profile $A$ is used.
The set of voters and the set of alternatives is omitted and chosen accordingly.
We use the $\oplus$ operator to concatenate two decision sequences.

\begin{theorem}\label{thm:setzerodsg}
\setzerodsg
\end{theorem}
\begin{proof}
Let us consider an arbitrary $(2n-2)$-decision sequence. We show (without loss of generality) that if voter $n$ does not win within the first $2n-3$ rounds, then $n$ necessarily wins in round $2n-2$.
To reflect that this decision sequence may be preceded by another decision sequence, let $a_1,\dots,a_n$ be arbitrary voter weights, and without loss of generality assume that $a_1\geq \dots \geq a_{n-1}$.
As the weight of voter $n$ increases in each round, her weight is at least $2n-2$ in round $2n-2$.
We will show that the total weight of voters $1,\dots, n-1$ in round $2n-2$ is less that $2n-2$ and hence voter $n$ is guaranteed to win in round $2n-2$, i.e., this proves that dry spells last for at most $2n-2$ rounds.

Let $s_t$ denote the weight that has been zeroed in round $t$, and
let $Z_t$ denote the set of voters that have been zeroed in rounds $1$ to $t$.
As a first step, we are going to prove for all $t\leq 2n-2$ that
\begin{align}
\sum_{j=1}^{t} s_j \geq \sum_{i\in Z_t\cup\{1,\dots,t\}} a_i + \frac{t(t+1)}{2}.\label{eq:perp-setzero-1}
\end{align}
We prove Ineqation~\eqref{eq:perp-setzero-1} by induction. For $t=0$ the statement obviously holds.
Assume that the statement holds for $t$.
We distinguish two cases: (i) $Z_{t+1} = Z_{t}$, and (ii) $Z_{t+1} \neq Z_{t}$.

Case (i): 
We know that voter $n$ does not win in round $t+1$. As the weight of voter $n$ in this round is at least $t+1$, the winning voters have a total weight of at least $t+1$.
Hence $s_{t+1}\geq t+1$. As we know that $Z_{t+1}=Z_t$, it follows that
\begin{align*}
\sum_{j=1}^{t+1} s_j  &\geq \sum_{i\in Z_t} a_i + \frac{t(t+1)}{2} +  s_{t+1}\\
 &\geq \sum_{i\in Z_{t+1}} a_i + \frac{t(t+1)}{2} +(t+1)\\
 &= \sum_{i\in Z_{t+1}} a_i + \frac{(t+1)(t+2)}{2} .
\end{align*}
Case (ii): 
Let $x\in Z_{t+1} \neq Z_{t}$. Voter $x$ has been zeroed in round $t+1$ for the first time. Hence, $x$'s voter weight is at least $a_x+t+1$. It follows that $s_{t+1} \geq \sum_{i\in Z_{t+1}\setminus Z_t} a_i +  (t+1)$
\begin{align*}
\sum_{j=1}^{t+1} s_j  &\geq \sum_{i\in Z_t} a_i + \frac{t(t+1)}{2} +   s_{t+1}\\
 &\geq \sum_{i\in Z_{t}} a_i + \frac{t(t+1)}{2} + \sum_{i\in Z_{t+1}\setminus Z_t} a_i +  (t+1)\\
 &= \sum_{i\in Z_{t+1}} a_i + \frac{(t+1)(t+2)}{2} .
\end{align*}

Towards a contradiction assume that $Z_{2n-3}\neq \{1,\dots, n-1\}$. By assumption $n\notin Z_{2n-3}$, so consequently there are at least two voters that have not been zeroed in rounds $1$ to $2n-3$; let the other voter be $k$.
We are now considering the voter weights in round $2n-3$.
We can bound the weight of voters $ \{1,\dots,k-1,k+1,\dots,n-1\}$ (using Inequation~\ref{eq:perp-setzero-1}) as follows:
\begin{align*}
&\underbrace{\sum_{i\in \{1,\dots,k-1,k+1,\dots,n-1\}} a_i}_{\text{starting weight}} + \underbrace{(n-2)(2n-3)}_{\text{the weight gained in $(2n-3)$ rounds}}\\ &- \underbrace{\sum_{t=1}^{2n-4} s_t}_{\text{zeroed weight in previous rounds}} \leq \\
&\leq (n-2)(2n-3) - \frac{(2n-4)(2n-3)}{2} = 0.
\end{align*}
As the weight of voter $k$ is exactly $a_k+2n-3$, we obtain a contradiction towards the assumption that voter $k$ did not win in this round. Hence, $Z_{2n-3} = \{1,\dots,n-1\}$.

It remains to show that voter $n$ necessarily has to win in round $2n-2$, and thus $n$'s dry spell can last at most $2n-3$ rounds.
We bound the weight of voters $ \{1,\dots,n-1\}$, as before by using Inequation~\ref{eq:perp-setzero-1}:
\begin{align*}
&\underbrace{\sum_{i\in \{1,\dots,n-1\}} a_i}_{\text{starting weight}} + \underbrace{(n-1)(2n-2)}_{\text{the weight gained in $(2n-3)$ rounds}} \\ &- \underbrace{\sum_{t=1}^{2n-3} s_t}_{\text{zeroed weight in previous rounds}} \leq \\
&\leq (n-1)(2n-2) - \frac{(2n-3)(2n-2)}{2} = n-1.
\end{align*}
Since the weight of voter $n$ is $a_n + 2n-2>n-1$, she must win this round.

\begin{table*}
\begin{center}
\begin{tabular}{cccccccccc}\toprule
• & 1 & 2 & 3 & $\dots$ & $n-3$ & $n-2$ & $n-1$ & $n$ & choice \\ 
\midrule
$P_1$ & $\{c_1\}$ & $\{c_2\}$ & $\{c_3\}$ & $\dots$ & $\{c_{n-3}\}$ & $\{c_{n-2}\}$ & $\{c_{n-1}\}$ & $\{c_{n}\}$ & $c_{1}$ \\ 
$P_2$ & $\{c_1\}$ & $\{c_2\}$ & $\{c_3\}$ & $\dots$ & $\{c_{n-3}\}$ & $\{c_{n-2}\}$ & $\{c_{n-1}\}$ & $\{c_{n}\}$ & $c_{2}$ \\ 
$\dots$ & • & • & • & • & • & • & • & • \\ 
$P_{n-1}$ & $\{c_1\}$ & $\{c_2\}$ & $\{c_3\}$ & $\dots$ & $\{c_{n-3}\}$ & $\{c_{n-2}\}$ & $\{c_{n-1}\}$ & $\{c_{n}\}$ & $c_{n-1}$ \\ 
$P_{n}$ & $\{c_1\}$ & $\{c_2\}$ & $\{c_3\}$ & $\dots$ & $\{c_{n-3}\}$ & $\{c_{n-2}\}$ & $\{c_{1}\}$ & $\{c_{n}\}$ & $c_{1}$ \\ 
$P_{n+1}$ & $\{c_1\}$ & $\{c_1\}$ & $\{c_3\}$ & $\dots$ & $\{c_{n-3}\}$ & $\{c_{n-2}\}$ & $\{c_{1}\}$ & $\{c_{n}\}$ & $c_{1}$ \\ 
$\dots$ & • & • & • & • & • & • & • & $\dots$ \\ 
$P_{2n-4}$ & $\{c_1\}$ & $\{c_1\}$ & $\{c_1\}$ & $\dots$ & $\{c_{1}\}$ & $\{c_{n-2}\}$ & $\{c_{1}\}$ & $\{c_{n}\}$ & $c_{1}$ \\ 
$P_{2n-3}$ & $\{c_1\}$ & $\{c_1\}$ & $\{c_1\}$ & $\dots$ & $\{c_{1}\}$ & $\{c_{1}\}$ & $\{c_{1}\}$ & $\{c_{n}\}$ & $c_{1}$ \\ 
$P_{2n-2}$ & $\{c_1\}$ & $\{c_1\}$ & $\{c_1\}$ & $\dots$ & $\{c_{1}\}$ & $\{c_{1}\}$ & $\{c_{1}\}$ & $\{c_{n}\}$ & $c_{n}$ \\\bottomrule 
\end{tabular} 
\caption{Sequence of approval profiles for lower bound construction in Theorem~\ref{thm:setzerodsg}}\label{tab:dry-spell-zero-profiles}
\end{center}
\end{table*}

\begin{table*}
\begin{center}
\begin{tabular}{ccccccccccc}\toprule
• & 1 & 2 & 3 & $\dots$ & $n-3$ & $n-2$ & $n-1$ & $n$ & \\ 
\midrule
$P_1$ & $\mathbf{1}$ & $1$ & $1$ & $\ldots$ & $1$ & $1$ & $1$ & $1$ &  \\ 
$P_2$ & $1$ & $\mathbf 2$ & $2$ & $\dots$ & $2$ & $2$  &  $2$ & $2$ &  \\ 
$\dots$ & • & • & • & • & • & • & • & • \\ 
$P_{n-1}$ & $n-2$ & $n-3$ & $n-4$ & $\dots$ & $2$ & $1$ & $\mathbf{n-1}$ & $n-1$ &  \\ 
$P_{n}$ & $\mathbf{n-1}$ & $n-2$ & $n-3$ & $\dots$ & $3$ & $2$ & $\mathbf 1$ & $n$ &  \\ 
$P_{n+1}$ & $\mathbf{1}$ & $\mathbf{n-1}$ & $n-2$ & $\dots$ & $4$ & $3$ & $\mathbf 1$ & $n+1$ &  \\ 
$\dots$ & • & • & • & • & • & • & • & • \\ 
$P_{2n-4}$ & $\mathbf{1}$ & $\mathbf{1}$ & $\mathbf{1}$ & $\dots$ & $\mathbf{n-1}$ & $n-2$ & $\mathbf 1$ & $2n-4$ &  \\ 
$P_{2n-3}$ & $\mathbf{1}$ & $\mathbf{1}$ & $\mathbf{1}$ & $\dots$ & $\mathbf{1}$ & $\mathbf{n-1}$ &  $\mathbf 1$ & $2n-3$ &  \\ 
$P_{2n-2}$ & ${1}$ & ${1}$ & ${1}$ & $\dots$ & ${1}$ & $1$ &  $1$ & $\mathbf{2n-2}$ &  \\ \bottomrule
\end{tabular} 
\caption{Voter weights for lower bound construction in Theorem~\ref{thm:setzerodsg}}\label{tab:dry-spell-zero-weights}
\end{center}
\end{table*}

To see that a better guarantee is not possible, for each $n>2$ we construct a $(2n-2)$-sequence of approval profiles as shown in Table~\ref{tab:dry-spell-zero-profiles}. The corresponding sequence of voter weights is shown in Table~\ref{tab:dry-spell-zero-weights}. We assume that ties are broken in order $c_1>c_2>\dots>c_n$, i.e., always disadvantaging voter $n$. Observe that voter $n$ does not win within the first $2n-3$ rounds, thus showing tightness of the bound.
\end{proof}

\begin{theorem}\label{thm:consensusdsg}
\consensusdsg
\end{theorem}
\begin{proof}
Let $(N, \bar A,\bar C,\bar w)$ be a $k$-decision history with $\bar w$ being the choice sequence generated by Perpetual Consensus.
Let $t$ be so that $t+\frac{n^2+3n}{4}\leq k$; we denote $t+\frac{n^2+3n}{4}$ with $t'$.
Then we have to show that for any voter $v$, $\satisfaction(v, \bar w_{\leq t})<\satisfaction(v, \bar w_{\leq t'})$.
If $\satisfaction(v, \bar w_{\leq t})<\satisfaction(v, \bar w_{\leq t'-1})$, then the statement holds.
Otherwise, let us consider round $t'$ and show that voter $v$ is guaranteed to be satisfied with the choice of Perpetual Consensus.

First, note that $v$ might have been been the only voter satisfied with the choice in round $t$.
Thus, $\alpha_{t+1}(v)$ might be negative; indeed it can be as low as $-n+2$. (In order for $v$ to win in round $t$, 
$\alpha_{t}(v)$ must be positive.)
Consequently, \[\alpha_{t'}(v)\geq -n+1+\frac{n^2+3n}{4}=\frac{n^2-n+4}{4}.\]

Secondly, the sum over all weights each rounds has to add up to $n$.
In particular, $\sum_{u\in N} \alpha_{t'}(u)=n$.
Let $N^+$ denote the voters in round $t'$ with positive weight (which includes $v$) and $N^-$ those with weight $\leq 0$.
Then \[\sum_{u\in N^+} \alpha_{t'}(u)-n = - \sum_{u\in N^-} \alpha_{t'}(u).\]
We try to bound this negative component $\sum_{u\in N^-} \alpha_{t'}(u)$.
Observe that each round the negative weights can decrease by at most $n-1$.
Further each negative weight increases by $1$ each round.
Let us consider how much negative weight can still exist in round $t'$ from previous rounds.
From round $t'-1$, a contribution of $-n+2$ is possible.
From round $t'-2$, a contribution of $-n+3$ is possible (assuming that only one negative voter has the maximum negative weight of $-n+1$ in round $t'-2$, which increased by $1$).
This continues until round $t'-n+1$, where the maximum possible negative weight of $-n+1$ has already been offset by the increase in the meantime.
Thus, we can bound \[\sum_{u\in N^-} \alpha_{t'}(u)\geq \sum_{i=1}^n (-n+i)=\frac{1}{2}(n^2-3n+2)\]
and \[\sum_{u\in N^+} \alpha_{t'}(u)\leq \frac{1}{2}(n^2-3n+2)+n=\frac{1}{2}(n^2-n+2)\]

Since \[\alpha_{t'}(v)\geq \frac{n^2-n+4}{4} > \frac{1}{2} \sum_{u\in N^+} \alpha_{t'}(u),\]
voter $v$ has more than half of the total weight and thus decides round $t'$.
\end{proof}

\begin{proposition}\label{prop:dryspellinfty}
\dryspellinfty
\end{proposition}

\begin{proof}
We assume that ties are broken in alphabetic order.
\begin{itemize}
\item AV: Consider $(\{a\},\{a\},\{b\})^k$. Voter $3$ never wins. Hence, AV has unbounded dry spells.

\item Perpetual PAV: Consider the $3k$-decision sequence
\begin{align*}
\left((\{a\},\{b\},\{a\})^1\oplus (\{a\},\{b\},\{b\})^1\right)^k \oplus(\{a\},\{a\},\{c\})^k
\end{align*}
The choice sequence determined by Perpetual PAV is \[(\underbrace{a,b,a,b,\dots,a,b}_{\text{$2k$ many}},\underbrace{a,\dots,a}_{\text{$k$ many}}).\]
Note that voter 3 does not win in rounds $2k+1$ to $3k$. As $k$ is unbounded, so is the potential length of dry spells.

\end{itemize}
\end{proof}

\begin{theorem}\label{thm:thmiud}
\thmiud
\end{theorem}

\begin{proof}
\leavevmode
\begin{itemize}
\item AV is loss-based and hence satisfies independence of uncontroversial decisions.

\item 
Perpetual Consensus satisfies IUD. It is easy to see that if round $k+1$ is an uncontroversial decision, then for all $v\in N$ it holds that $\alpha_{k+2}(v)=\alpha_{k+1}(v)$.
\end{itemize}
\end{proof}

\begin{proposition}\label{prop:propiud}
\propiud
\end{proposition}

\begin{proof}
\leavevmode
\begin{itemize}
\item Perpetual PAV: Consider $(\{a\},\{a\},\{a\},\{b\})^4$. The corresponding decision sequence is $(a,a,a,b)$. If we insert the approval profile $(\{c\},\{c\},\{c\},\{c\})$ before the first round, we obtain the sequence $(c,a,a,a,a)$. (Note that in the last round the weight of $a$ is $\frac{3}{5}$ and the weight of $b$ is $\frac{1}{2}$; hence $a$ wins.) Thus IUD is not satisfied.
\item Perpetual Reset: Consider $(\{a\},\{a\},\{b\})^3$. The corresponding decision sequence is $(a,a,b)$. If we insert the approval profile $(\{c\},\{c\},\{c\})$ before the third round, we obtain the sequence $(a,a,c,a)$. Thus IUD is not satisfied.
\item Rotating Dictator: Rotating Dictator does not satisfy independence of uncontroversial decision
by definition.
\end{itemize}
\end{proof}

\begin{proposition}\label{prop:failsp}
\failsp
\end{proposition}
\begin{proof}
We assume that ties are broken in alphabetic order. Observe that all $k$-decision sequences in counterexamples are simple and have $k=|N|$.
\begin{itemize}
\item For $(\{a\},\{a\},\{b\})^3$, AV outputs the choice sequence $(a,a,a)$. Voter $3$ never wins.
\item Perpetual Reset: Consider $(\{a\},\{a\},\{a\},\{b\},\{c\},\{d\})^{6}$. The corresponding choice sequence is $ \bar w=(a,a,a,b,a,c)$ and, considering for example voter $1$, we see that $\satisfaction(1, \bar w_{\leq 6})=4\neq3=\quota_6(1)$.
\end{itemize}
\end{proof}

\begin{theorem}\label{thm:satisfysp}
\satisfysp
\end{theorem}
\begin{proof}
Let us consider each of the three rules in turn:
\paragraph{Perpetual PAV:} 
This follows from the fact that Sequential PAV (which can be seen as a special case of Perpetual PAV) is identical to the D'Hondt method on party-list profiles~\citep{brill2018multiwinner,jet-consistentabc}. It hence satisfy the lower quota axiom from the apportionment literature~\cite{Balinski1982FairRepresentation:Meeting}.
It is straightforward to see that lower quota implies simple proportionality.

\paragraph{Perpetual Consensus:} Towards a contradiction, assume that Perpetual Consensus outputs for a simple $n$-decision sequence $(N, \bar A,\bar C)$ with $|N|=n$ a $k$-choice sequence $\bar w$ that is not proportional.
Since Perpetual Consensus never chooses alternatives that are not approved by any voter,
there has to be at least one voter $v\in N$ with $\satisfaction(v, \bar w)>\quota_n(v)$.
Let $\quota_n(v)=a$ and consider the round $i\leq n$ where voter $v$ increased her satisfaction to $a+1$.

Recall that there are $a$ many voters with $v$'s preference; let us name them $N_v$.
The total weight of voters in $N_v$ in round $i$ is \[\sum_{v'\in N_v} \alpha_i(v') = a\cdot i - n\cdot a = a\cdot (i-n) \leq 0,\] as $\satisfaction(v', \bar w_{\leq i})=a$.
Since the total weight of all voters is $n$, there is at least one voter with positive weight.
Thus, we have obtained a contradiction to the fact that in round $i$ the satisfaction of $v$ increased, i.e., $i$'s preferred alternative won.

\end{proof}

\end{document}